\documentclass[11pt]{article}

\usepackage[T1]{fontenc}
\usepackage[utf8]{inputenc}
\usepackage{amsmath,amsthm}
\usepackage{newtxtext,newtxmath}
\usepackage{microtype}
\usepackage[margin=1.12in]{geometry}
\usepackage{booktabs}
\usepackage{graphicx}
\usepackage{tikz}
\usetikzlibrary{positioning}
\usepackage{pgfplots}
\pgfplotsset{compat=1.18,
  every axis/.append style={
    line width=0.7pt, tick style={line width=0.5pt},
    label style={font=\small}, tick label style={font=\scriptsize},
    legend style={fill=white, fill opacity=0.85, text opacity=1,
                  draw=black!25, font=\scriptsize}}}
\usepackage{algorithm}
\usepackage{algpseudocode}
\usepackage{float}
\usepackage[section]{placeins}
\usepackage{flafter}
\usepackage[font=small,labelfont=bf,skip=6pt]{caption}
\usepackage[hidelinks]{hyperref}

\newtheorem{theorem}{Theorem}[section]
\newtheorem{lemma}[theorem]{Lemma}
\newtheorem{proposition}[theorem]{Proposition}
\newtheorem{corollary}[theorem]{Corollary}
\theoremstyle{definition}
\newtheorem{definition}[theorem]{Definition}

\newtheorem{conjecture}[theorem]{Conjecture}
\theoremstyle{remark}
\newtheorem{remark}[theorem]{Remark}

\newcommand{\rk}{r}
\newcommand{\est}{\hat{r}}
\newcommand{\E}{\mathbb{E}}
\newcommand{\Prob}{\mathbb{P}}
\newcommand{\N}{\mathbb{N}}
\newcommand{\diff}{\,\triangle\,}

\title{Comparison Patrols on Drifting Orders:\\
Certified Rank Maintenance, Evolving Planar Maxima,\\
and Selection under Drifting Fitness}

\author{Faruk Alpay\thanks{Correspondence: \texttt{alpay@lightcap.ai}}
\quad Levent Sar{\i}o\u{g}lu\\[4pt]
\small Department of Computer Engineering, Bah\c{c}e\c{s}ehir University,
Istanbul, Turkey\\
\small \texttt{\{faruk.alpay, levent.sarioglu\}@bahcesehir.edu.tr}}

\date{}

\begin{document}
\maketitle

\begin{abstract}
Rank-based selection in a changing environment acts on a ranking that
decays while it is being used: tournaments judged on stale fitness advance
the wrong parents, elitism protects fallen champions, nondominated fronts
are misreported, and re-evaluating everyone competes with search for the
same budget.  This paper gives that information layer exact foundations.
A hidden total order on $n$ items drifts by Poisson adjacent
transpositions while a structure asks one truthful pairwise comparison per
step and serves, at every instant, not a ranking but a ranking \emph{with
promises}: deterministic verification ages, per-item displacement
certificates, and transfer guarantees for every selection rule downstream,
in $O(1)$ per operation from $3n+O(1)$ words.  Every
oblivious probing schedule suffers expected Kendall error at least
$\min(\alpha,1)(n-1)/24$ beyond a short horizon, and we prove a floor of the
same order, with explicit constants at every rate, for the
location-oblivious class containing the structure itself, so its
equilibrium is $\Theta(\min(\alpha,1)\,n)$ with both sides exhibited.  A
bump lemma makes the comparison patrol exactly self-stabilizing: at zero
drift it sorts an estimate of maximum overstatement $L$ in at most $L$
cycles and in more than $L-1$, so an abrupt shock of overstatement $L$
costs at most $(L{+}1)(n{-}1)$ probes against the
$\sum_{m\le n}\lceil\log_2 m\rceil$ of a full rebuild: a provable policy
crossover at $L\approx\log_2 n$ that a swap-count detector exploits, within
twice the better pure response, without being told $L$.  At equilibrium an
exact birth and death ledger of discordant pairs verifies Little's law to
within one percent; the steady-state constant it supports, one half, is
stated as a conjecture and kept explicitly separate from the proven floor,
with every claim in the paper carrying its status, proved, measured, or
conjectured, on its face.  Certified rank queries return a two-part
contract, a proven motion radius plus a calibrated residual with
finite-sample validity.
Deterministic transfer theorems carry the maintained order into evolving
planar maxima, with a localized charging bound counting only
frontier-incident discordances, and into truncation, tournament, elitist,
and two-objective Pareto selection under drifting fitness:
$|T_k\diff\widehat T_k|\le2\lfloor\sqrt K\rfloor$, tournament error exactly
$K/\binom n2$, $|M\diff\widehat M|\le2(K^x+K^y)$.  A thirty-seed study
scales to $n=65{,}536$, audits the certificates, measures the recovery
laws, and embeds eleven re-evaluation policies, including random
immigrants, triggered hypermutation, memory, restart, and a
patrol+refresh hybrid, in equal-budget dynamic evolutionary loops with
frequency, budget, and selection-pressure ablations, locating exactly
where certified local maintenance beats global re-evaluation and where
it must hand over.
\end{abstract}

\section{Introduction}\label{sec:intro}

Every rank-based selection rule in a drifting environment trusts an
object that nobody maintains.  Dynamic evolutionary algorithms advance
parents by tournaments, keep top segments, protect elites, and carry
nondominated fronts, all read off a fitness ranking whose entries were
measured at different times; the environment keeps moving between
measurements.  Re-evaluating the population detects the change but spends
the evaluations that would have generated offspring; trusting the cache
spends nothing and silently corrupts every decision downstream, advancing
the wrong parents and shielding champions that have already fallen.  The
question this paper starts from is therefore an evolutionary one:
\emph{what can selection still rely on when current fitness comparisons,
rather than arithmetic, are the scarce resource?}

Our answer is to build the missing object as a data structure and prove
exactly what it can promise.  A hidden total order on $n$ items drifts by
elementary moves while a structure spends one truthful pairwise
comparison per step and keeps serving rank queries anyway.  The
deliverable we ask of it is deliberately more than an estimate.  It is an
estimate \emph{with promises}: a deterministic bound on how stale every
entry can be, a certificate attached to every reported rank that bounds
the hidden motion since that entry was last verified, exact recovery laws
after abrupt changes, and transfer theorems that carry these promises
into whatever consumes the ranking.  Truncation, tournament, elitist, and
Pareto selection consume exactly this object, and the guarantees survive
verbatim because they never reference what the items are.  The central
technical question is one of maintained guarantees: \emph{what can a unit
comparison budget promise, per item and per query, about an order it can
never see at rest?}

The split separates three layers that the dynamic-optimization
literature tends to blur, and the paper holds the boundary explicitly.
\emph{Rank-process maintenance}, what the board itself can promise under
drift, is proved (Sections~\ref{sec:patrol} and~\ref{sec:cert}).
\emph{Selection-error transfer}, how board error becomes decision error
for truncation, tournaments, elitism, and Pareto fronts, is proved, with
a coupling proposition carrying per-decision rates into whole-trajectory
guarantees for any tournament-driven algorithm
(Sections~\ref{sec:maxima} and~\ref{sec:selection}).  \emph{Optimization
performance}, what those decisions earn on a particular landscape, is
empirical by design, and the full-scale loops of
Section~\ref{sec:exp-ea} measure it against the field's own response
repertoire at equal evaluation budgets, with paired tests and effect
sizes rather than bare orderings.

Our tractable baseline is a uniform adjacent-transposition process on a
Poisson clock, the elementary mutation kernel of the symmetric group read
as an adversary.  Proposition~\ref{prop:fitness-rank} shows why adjacent
swaps are the exact ranking-level primitive of continuous trajectories with
isolated crossings; it also makes clear what the abstraction loses.  The
Poisson and spatial-uniformity assumptions are used only where stated, and
the experimental study later replaces them by clustered crossings, hot
spots, rate regimes, and nonlocal shocks.

\emph{First: what is unavoidable?}  Theorem~\ref{thm:lower} shows that any
probing schedule fixed in advance, no matter how clever the estimator that
digests the answers, must tolerate expected Kendall distance
$\Omega(\min(\alpha,1)\,n)$ at every time beyond a horizon of $n/2$ steps,
with explicit constants ($\tfrac{n-1}{12}$ at unit rate).  The proof
technique is the technical heart of the paper: after Poissonization the
event counts at the $n-1$ rank locations become independent, and the parity
of the count at any single location is shown to be independent of the entire
transcript the algorithm ever sees, provided the location's two neighbors
stayed quiet and the resident pair was never probed directly.  Disorder then
survives not because the algorithm is slow but because the transcript
carries \emph{zero bits} about a linear number of coin flips.

\emph{Second: what structure should one run?}  We analyze the
\emph{comparison patrol} (Section~\ref{sec:patrol}): three arrays and a
cursor in $3n+O(1)$ words; the cursor sweeps the estimate's adjacent
positions cyclically, repairs each local discordance it meets, and
timestamps both items it touches.  The patrol does $O(1)$ work per step,
answers every query in $O(1)$, probes every item at least once per cycle,
and guarantees a deterministic \emph{verification age} of at most $2(n-1)$
for every item at every time, within a factor four of what any schedule
with the same budget can achieve (Proposition~\ref{prop:age}).  A
boustrophedon variant, which reverses direction at the ends in the manner
of coverage paths, trades a slightly better error constant for a
deterministic $4(n-1)$ age bound.

\emph{Third: what does the structure cost at equilibrium, and what after a
shock?}  A one-line invariant carries most of the new theory: in a
drift-free cycle, \emph{every} item whose reported rank exceeds its hidden
rank is swapped down exactly once (Lemma~\ref{lem:bump}).  Consequently the
patrol is self-stabilizing in the exact sense of needing at most $L$ and
more than $L-1$ cycles to sort an estimate whose maximum overstatement is
$L$ (Theorem~\ref{thm:stabilize}); the reversed board, the worst case, is
measured at $1022.001$ cycles against the predicted window
$(1022,\,1023]$ at $n=1024$.  The
same lemma prices abrupt change: a shock of overstatement $L$ is repaired
in at most $(L{+}1)(n{-}1)$ probes, a full binary-insertion rebuild costs
$\sum_{m=2}^{n}\lceil\log_2 m\rceil$ probes whatever happened, and the two
laws cross at $L\approx\log_2 n$ (Theorem~\ref{thm:recovery}).  A hybrid
that watches only its own swap counter, the one statistic the structure
owns for free, recovers within twice the better pure response without
being told $L$ (Proposition~\ref{prop:hybrid}).  At equilibrium we prove a
floor: every location-oblivious prober, the patrols included, holds
expected Kendall error at least $\min(\alpha,1)(n-1)/24$
(Theorem~\ref{thm:floor}), so the patrol equilibrium is
$\Theta(\min(\alpha,1)\,n)$ with explicit constants on both sides; the
measured ceiling constant $0.5$ is stated as a conjecture
(Conjecture~\ref{conj:balance}) and audited, rather than asserted, by an
exact birth and death ledger of discordant pairs that verifies Little's law
to within one percent (Section~\ref{sec:exp-little}).

\emph{Fourth: what can the structure promise per query?}  Every rank query
is answered in $O(1)$ together with a \emph{displacement certificate}
(Section~\ref{sec:cert}): an interval, computed from the item's
verification age alone, that bounds how far the hidden rank can have moved
since the item was last touched.  Soundness is Lemma~\ref{lem:tail}, a
Freedman-type martingale tail combined with a Poisson tail; the half-width
depends on the age only through the ratio $g/n$, so a patrol-maintained
board promises $\pm 8$ positions at the $95\%$ level for $n=1024$, and the
same width at $n=4096$.  The promise is served as an explicit two-part
contract (Definition~\ref{def:twopart}): a proven motion radius, valid
under any maintainer and any schedule, composed with a calibrated residual
whose finite-sample validity across independent calibration runs is itself
a proposition (Proposition~\ref{prop:calibration}), not a hope.  The
certificates are the kinetic notion of a certificate transplanted to a
setting with no trajectories: nothing is known about motion except its
statistics.

The same machinery then leaves the line and enters the plane
(Section~\ref{sec:maxima}).  A planar point set in general position is, up
to monotone reparametrization, a pair of total orders; when both orders
drift, the set of maximal points, the staircase frontier of
Kung, Luccio, and Preparata~\cite{KungLucPre75}, drifts with them.  Running
one patrol per coordinate maintains the frontier: we prove the purely
combinatorial bound
$|M \diff \widehat{M}| \le 2(K^x + K^y)$
relating the misreported frontier to the two Kendall distances
(Theorem~\ref{thm:maxima}), exhibit a family on which the bound is met up to
its factor two, and assemble the reported frontier in $O(n)$ on demand.
The charging argument localizes: only discordant pairs with both endpoints
on the union of the two frontiers can ever be charged
(Theorem~\ref{thm:maxima-local}), which replaces the global Kendall budget
by a frontier-incident one and closes most of the gap between the
worst-case bound and generic instances, measured across front geometries
from logarithmically sparse to fully antidiagonal
(Section~\ref{sec:exp-maxima}).

The maintained order is consumed directly by evolutionary selection
(Section~\ref{sec:selection}).  Items are candidates, the hidden order is the
current fitness ranking, and the patrol supplies truncation, tournament, and
elitist selection with explicit staleness information.  The connection is
carried by theorems rather than analogy.  Truncation selection
obeys the transfer bound
$|T_k \diff \widehat{T}_k| \le 2\lfloor\sqrt{K}\rfloor$
(Theorem~\ref{thm:topk}), and a block-exchange family attains the bound at
every value of its right-hand side; a binary tournament between a
uniformly random pair errs with probability exactly $K_t/\binom{n}{2}$
(Proposition~\ref{prop:tournament}); and the displacement certificates
upgrade to per-decision selection guarantees through a margin rule that
selects only items remaining inside the top $k$ under the certified worst
case (Proposition~\ref{prop:elitism}).  The frontier of
Section~\ref{sec:maxima} is the two-objective member of the same family:
the staircase is the Pareto frontier of two drifting objectives,
maintained on the same budget with the same certificates.

The empirical study has three levels (Section~\ref{sec:experiments}).  The
first treats the structure as a data structure: step time, memory, access
locality, certificate tightness, and the reversed-board stabilization law,
measured to $n=65{,}536$.  The second stresses the maintenance model:
thirty paired seeds, five drift rates, non-Poisson regimes, explicit local
and nonlocal shocks with their recovery laws and the hybrid's switching
behavior, a Little's law audit of the equilibrium through an exact
inversion ledger, a split-seed calibration audit of the two-part
certificate, and Pareto fronts ranging from logarithmically sparse to fully
antidiagonal.  The third embeds every policy in the same steady-state EA on
Dynamic BitMatching and Moving Peaks.  Patrol, periodic full re-evaluation,
random partial, elite-first, and sentinel-triggered re-evaluation, no
re-evaluation, a budget-free oracle, and the stronger dynamic-optimization
responses, random immigrants, triggered hypermutation, memory with
reinjection, partial restart, and the patrol+refresh hybrid, are compared
under identical total objective-evaluation budgets.  We report optimization
error, recovery, selection mistakes, Kendall disorder, certified yield, and
empirical diagnostics of the induced rank process with paired confidence
intervals and multiple-test correction.  The diagnostics close a loop the
theory opens: the recovery crossover at $L\approx\log_2 n$ becomes a
measurable decision rule, and the benchmark landscapes are placed on the
correct side of it.

\paragraph{Contributions.}
\begin{enumerate}
\item A fully specified maintained-order abstract data type under a unit
  comparison budget: $O(1)$ probe work, $O(1)$ certified queries,
  $3n+O(1)$ words, an exact unattended-decay envelope
  (Lemma~\ref{lem:decay}), and an explicit reduction from continuous
  dynamic fitness to adjacent rank crossings
  (Proposition~\ref{prop:fitness-rank}).
\item An $\Omega(\min(\alpha,1)\,n)$ steady-state lower bound for all
  oblivious schedules with explicit constants, by Poissonization and parity
  decoupling (Theorem~\ref{thm:lower}), and a companion floor of the same
  order for every location-oblivious prober, the patrols included
  (Theorem~\ref{thm:floor}): the patrol equilibrium is
  $\Theta(\min(\alpha,1)\,n)$ with both constants exhibited and the
  remaining gap, the ceiling constant $\approx 0.5$, isolated as a
  conjecture (Conjecture~\ref{conj:balance}) with an exact
  ledger-level audit (Section~\ref{sec:exp-little}).
\item The bump lemma and exact self-stabilization: a drift-free cycle
  swaps every overstated item down exactly once (Lemma~\ref{lem:bump}),
  so sorting from maximum overstatement $L$ takes at most $L$ and more
  than $L-1$ cycles (Theorem~\ref{thm:stabilize}).
\item A deterministic recovery calculus for abrupt change: shock
  transfer $K\mapsto K+J$ (Proposition~\ref{prop:shock}), quiet-aftermath
  recovery within $(L{+}1)(n{-}1)$ probes, a provable crossover against
  the $\sum_{m\le n}\lceil\log_2 m\rceil$ rebuild at $L\approx\log_2 n$
  (Theorem~\ref{thm:recovery}), and a swap-count hybrid within twice the
  better pure response without knowing $L$
  (Proposition~\ref{prop:hybrid}).
\item Displacement certificates with proven coverage, $n$-free width at
  fixed relative age, a non-homogeneous intensity extension
  (Lemma~\ref{lem:tail}, Corollary~\ref{cor:variable-rate}), and an
  explicit two-part query contract: proven motion radius composed with a
  calibrated residual carrying a finite-sample validity guarantee
  (Definition~\ref{def:twopart}, Propositions~\ref{prop:compose}
  and~\ref{prop:calibration}).
\item Deterministic, instance-wise frontier bounds: the global charging
  bound $|M \diff \widehat{M}| \le 2(K^x{+}K^y)$, tight up to its
  constant (Theorem~\ref{thm:maxima}), its frontier-local refinement
  charging only discordances inside $M\cup\widehat M$
  (Theorem~\ref{thm:maxima-local}), and $O(n)$ assembly.
\item A selection layer for drifting fitness: the truncation transfer
  bound $|T_k \diff \widehat{T}_k| \le 2\lfloor\sqrt{K}\rfloor$ with a
  matching family at every value of the bound (Theorem~\ref{thm:topk}),
  the exact tournament error identity
  (Proposition~\ref{prop:tournament}), and certified elitism with
  per-decision confidence (Proposition~\ref{prop:elitism}), priced
  against a generational re-evaluation baseline at equal budget
  (Section~\ref{sec:exp-selection}).
\item A full dynamic-EA evaluation on two benchmark families, eleven
  policies including random immigrants, triggered hypermutation, memory,
  restart, and the patrol+refresh hybrid, gradual and abrupt changes,
  three severities, equal objective-evaluation budgets, and induced
  rank-process diagnostics that instantiate the crossover of
  Theorem~\ref{thm:recovery} as a practical policy rule
  (Section~\ref{sec:exp-ea}).
\item A seeded, test-covered ancillary package reproducing every table and
  figure, including exact theorem tests, per-run CSV data, paired statistical
  tests, source hashes, and a one-command paper profile.
\end{enumerate}

\paragraph{The ledger of claims.}
The paper deliberately runs three kinds of statement side by side, and
Table~\ref{tab:ledger} is the boundary between them, stated once so that
no later sentence has to be read twice.  Everything labeled a theorem,
lemma, proposition, or corollary is proved in full; everything empirical
is a thirty-seed (or ten-seed, where stated) measurement with its
protocol in Section~\ref{sec:experiments}; exactly one statement is a
conjecture, it concerns a single constant, and both of its neighbors,
the floor below and the audit beside it, are exact.

\begin{table}[tb]
\centering\small
\caption{Status of every load-bearing claim.  Proven results hold with
the stated constants; tight means a matching family is constructed;
measured values are seeded experiments; one constant is conjectural.}
\label{tab:ledger}
\begin{tabular}{lll}
\toprule
claim & status & where\\
\midrule
oblivious error floor $\min(\alpha,1)\frac{n-1}{24}$
  & proven & Thm.~\ref{thm:lower}\\
patrol-class floor of the same order
  & proven & Thm.~\ref{thm:floor}\\
verification age $\le 2(n-1)$, factor-$4$ optimal
  & proven, tight & Prop.~\ref{prop:age}\\
stabilization window $(L-1,\,L]$ cycles
  & proven, exact & Thm.~\ref{thm:stabilize}\\
recovery crossover at $L \approx \log_2 n$
  & proven (quiet aftermath) & Thm.~\ref{thm:recovery}\\
hybrid within $2\,\mathrm{min} + 2(n-1)$
  & proven (quiet aftermath) & Prop.~\ref{prop:hybrid}\\
motion certificate coverage $\ge 1-\delta$
  & proven & Lem.~\ref{lem:tail}\\
residual padding $b$
  & calibrated, split-seed & Prop.~\ref{prop:calibration}\\
frontier bounds, global and localized
  & proven, tight family & Thms.~\ref{thm:maxima},~\ref{thm:maxima-local}\\
selection transfer $2\lfloor\sqrt K\rfloor$, tournament identity
  & proven, tight & Thm.~\ref{thm:topk}, Prop.~\ref{prop:tournament}\\
equilibrium constant $K^\ast/(\alpha n) \to \tfrac12$
  & \emph{conjectured}, audited & Conj.~\ref{conj:balance},
    \S\ref{sec:exp-little}\\
equilibrium constants $0.508$--$0.554$; EA orderings
  & measured & \S\ref{sec:experiments}\\
\bottomrule
\end{tabular}
\end{table}

The results form three short, independent chains plus two standalone
floors, and Figure~\ref{fig:depgraph} draws the map: the repair chain
(bump lemma to self-stabilization to recovery to the hybrid), the
certificate chain (motion tail to two-part contract to certified
elitism to failure time to the coupling transfer), and the charging
chain (global to localized to its necessity), with the two lower bounds
and the transfer identities standing on the model alone.  No proof in
the paper depends on the conjecture.

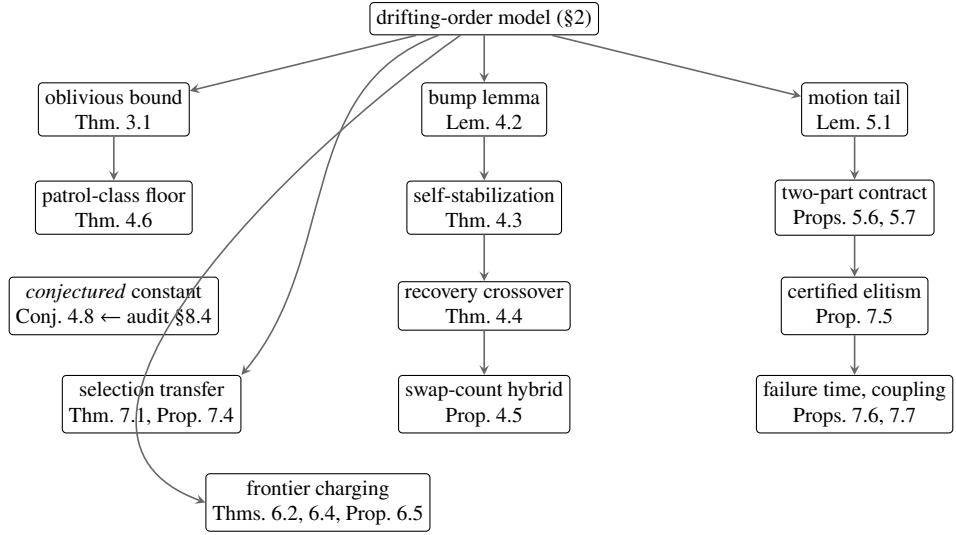
\begin{figure}[tb]
\centering
\begin{tikzpicture}[
  every node/.style={draw, rounded corners=1.5pt, font=\scriptsize,
                     inner sep=2.5pt, align=center},
  arr/.style={->, >=stealth, line width=0.6pt, draw=black!60}]
\node (model) at (0, 0) {drifting-order model (\S\ref{sec:model})};
\node (lower) at (-4.9, -1.2)
  {oblivious bound\\Thm.~\ref{thm:lower}};
\node (floor) at (-4.9, -2.5)
  {patrol-class floor\\Thm.~\ref{thm:floor}};
\node (conj) at (-4.9, -3.8)
  {\emph{conjectured} constant\\Conj.~\ref{conj:balance}
   $\leftarrow$ audit \S\ref{sec:exp-little}};
\node (bump) at (0, -1.2) {bump lemma\\Lem.~\ref{lem:bump}};
\node (stab) at (0, -2.5)
  {self-stabilization\\Thm.~\ref{thm:stabilize}};
\node (rec) at (0, -3.8) {recovery crossover\\Thm.~\ref{thm:recovery}};
\node (hyb) at (0, -5.1) {swap-count hybrid\\Prop.~\ref{prop:hybrid}};
\node (tail) at (4.9, -1.2)
  {motion tail\\Lem.~\ref{lem:tail}};
\node (twopart) at (4.9, -2.5)
  {two-part contract\\Props.~\ref{prop:compose},~\ref{prop:calibration}};
\node (elit) at (4.9, -3.8)
  {certified elitism\\Prop.~\ref{prop:elitism}};
\node (fail) at (4.9, -5.1)
  {failure time, coupling\\Props.~\ref{prop:failure-time},%
   ~\ref{prop:coupling}};
\node (sel) at (-4.4, -5.1)
  {selection transfer\\Thm.~\ref{thm:topk},
   Prop.~\ref{prop:tournament}};
\node (charge) at (-2.2, -6.4)
  {frontier charging\\Thms.~\ref{thm:maxima},~\ref{thm:maxima-local},
   Prop.~\ref{prop:local-necessity}};
\draw[arr] (model) -- (lower);
\draw[arr] (model) -- (bump);
\draw[arr] (model) -- (tail);
\draw[arr] (lower) -- (floor);
\draw[arr] (bump) -- (stab);
\draw[arr] (stab) -- (rec);
\draw[arr] (rec) -- (hyb);
\draw[arr] (tail) -- (twopart);
\draw[arr] (twopart) -- (elit);
\draw[arr] (elit) -- (fail);
\draw[arr] (model) to[out=200, in=50] (sel.north east);
\draw[arr] (model) to[out=215, in=160] (charge.west);
\end{tikzpicture}
\caption{The dependency map.  Three proof chains and two floors stand
on the model; the conjecture (bottom left) supports nothing, it is
itself supported, by the floor beneath it and the ledger audit beside
it.}
\label{fig:depgraph}
\end{figure}

\section{The drifting-order model}\label{sec:model}

\subsection{Process, probes, estimates}

Let $V$ be a set of $n \ge 3$ items.  At every time $t \in \N$ the hidden
state is a bijection $\rk_t : V \to \{1,\dots,n\}$, the \emph{hidden
ranking}; $\rk_t(x)$ is the rank of item $x$, rank $1$ being least.  Ranks
$\ell$ and $\ell+1$ are bridged by \emph{location} $\ell \in
\{1,\dots,n-1\}$.

Time step $t$ consists of two phases.
\begin{enumerate}
\item \textbf{Drift.}  $N_t \sim \mathrm{Poisson}(\alpha)$ events execute,
  independently of everything else.  Each event draws a location $\ell$
  uniformly from $\{1,\dots,n-1\}$ and swaps the two items currently at
  ranks $\ell$ and $\ell+1$.
\item \textbf{Probe.}  The algorithm selects a pair $\{x,y\} \subseteq V$
  and receives the truthful bit $[\rk_t(x) < \rk_t(y)]$, where $\rk_t$ is
  the ranking after the drift phase.
\end{enumerate}
The drift kernel is the elementary mutation operator on the symmetric
group; a single event is the smallest change a total order can absorb.  The
algorithm maintains an \emph{estimate} $\est_t : V \to \{1,\dots,n\}$, also
a bijection, which it may revise after each probe.  A \emph{schedule} is the
rule selecting the probed pairs.  A schedule is \emph{oblivious} if the
entire sequence of probed pairs is fixed before time begins (possibly drawn
at random, independently of the drift); it is \emph{adaptive} if probes may
depend on earlier answers.  Estimators are never restricted: they may be
adaptive, randomized, and computationally unbounded.  The patrol of
Section~\ref{sec:patrol} is adaptive in the mildest possible way (its
probe targets depend on its own estimate) and uses $O(1)$ work per step.

The two stochastic assumptions play different roles.  The Poisson clock
makes event counts on disjoint intervals independent and supplies the tails
used in Section~\ref{sec:cert}.  Uniform choice among the $n-1$ rank locations
makes elementary moves spatially exchangeable and determines the martingale
variance.  By contrast, the shock, selection, and Pareto-front transfer
results are deterministic and use neither assumption.  The robustness study
of Section~\ref{sec:exp-robust} deliberately violates both.

\begin{table}[t]
\centering\small
\caption{Rank orientation.  Larger ranks always mean fitter candidates or
larger coordinates; the implementation uses the same orientation with
zero-based arrays.}
\label{tab:notation}
\begin{tabular}{lll}
\toprule
object & manuscript & ancillary code\\
\midrule
least fit & rank $1$ & rank $0$\\
most fit & rank $n$ & rank $n-1$\\
top-$k$ boundary & $r(x)>n-k$ & $r[x]\ge n-k$\\
estimated rank & $\est_t(x)$ & \texttt{pos[x]}\\
true rank & $\rk_t(x)$ & \texttt{rank[x]}\\
\bottomrule
\end{tabular}
\end{table}

\subsection{Disorder functionals}

For two rankings $\sigma, \tau$ on $V$ the \emph{Kendall distance}
$K(\sigma,\tau)$ is the number of discordant pairs, i.e.\ pairs
$\{x,y\}$ ordered oppositely by $\sigma$ and $\tau$, and the
\emph{footrule} $F(\sigma,\tau) = \sum_{x \in V} |\sigma(x)-\tau(x)|$ is
the total rank displacement.  Diaconis and Graham~\cite{DiaGra77} proved
\begin{equation}\label{eq:dg}
K(\sigma,\tau) \;\le\; F(\sigma,\tau) \;\le\; 2\,K(\sigma,\tau),
\end{equation}
so the two functionals agree up to a factor two and $F/n$ reads as the
\emph{mean per-item rank error}.  We abbreviate
$K_t = K(\est_t, \rk_t)$ and $F_t = F(\est_t, \rk_t)$.

Two more quantities are bookkept by the structures of this paper.  The
\emph{verification age} $a_t(x)$ of item $x$ is the number of steps since a
probe last involved $x$ (items count as verified at $t=0$).  A
\emph{displacement certificate} for $x$ is a radius $D$ such that the
hidden rank of $x$ has moved less than $D$ since $x$'s last verification,
with prescribed confidence; certificates are the subject of
Section~\ref{sec:cert}.

\subsection{Unattended decay}

The first fact of the model is that an estimate left alone rots at unit
rate, and the onset of the rot is exactly computable.  This curve is the
zero-budget envelope against which every maintainer is judged.

\begin{lemma}[unattended decay]\label{lem:decay}
Fix any estimate $\sigma$ and let the hidden ranking start at
$\rk_0 = \sigma$ and drift.  After $m$ events,
\[
\E\,K(\sigma,\rk) \;\ge\; \frac{n-1}{2}\Bigl(1 - \rho^{\,m}\Bigr),
\qquad \rho = 1 - \tfrac{2}{n-1},
\]
and consequently, at time $t$ under the Poisson clock,
\[
\E\,K_t \;\ge\; \frac{n-1}{2}\Bigl(1 - e^{-2\alpha t/(n-1)}\Bigr).
\]
In particular $\E K_t \ge \tfrac{1}{2}\alpha t\,(1 - \tfrac{\alpha
t}{n-1})$ for all $t$: disorder accrues at essentially one unit per event
until it saturates at the $\Theta(n)$ floor.
\end{lemma}

\begin{proof}
Condition on the configuration after $m$ events and let $d$ be the number
of locations whose two current residents form a discordant pair with
respect to $\sigma$.  Distinct locations hold distinct pairs, so $d \le
K(\sigma,\rk)$.  The next event increases $K$ by one if it hits any of the
other $n-1-d$ locations and decreases it by one otherwise, whence
\[
\E\bigl[\Delta K \,\big|\, \text{configuration}\bigr]
 \;=\; 1 - \frac{2d}{n-1} \;\ge\; 1 - \frac{2K}{n-1}.
\]
Writing $a_m = \E K$ after $m$ events and taking expectations,
$a_{m+1} \ge \rho\,a_m + 1$ with $a_0 = 0$, and induction gives
$a_m \ge \tfrac{n-1}{2}(1-\rho^m)$ exactly (the affine map has fixed point
$\tfrac{n-1}{2}$).  The number of events in $t$ steps is
$N \sim \mathrm{Poisson}(\alpha t)$ and is independent of their locations,
so $\E K_t \ge \tfrac{n-1}{2}(1 - \E\rho^{\,N})
 = \tfrac{n-1}{2}(1 - e^{-\alpha t(1-\rho)})$,
which is the stated bound.  The final inequality uses
$1-e^{-z} \ge z - z^2/2$.
\end{proof}

\begin{remark}\label{rem:mixing}
The lemma bounds only the onset.  Left alone forever, the hidden ranking
mixes to uniform, at expected Kendall distance $n(n-1)/4$ from any fixed
estimate, on the $\Theta(n^3 \log n)$ timescale of the adjacent
transposition shuffle established by Wilson~\cite{Wilson04}.  The
maintenance problem lives three orders of magnitude below that horizon:
the question is never whether information dies, but how much of it a unit
probe budget can keep alive per step.
\end{remark}

\subsection{From dynamic fitness to rank drift}\label{sec:fitness-model}

The connection to dynamic evolutionary optimization is exact for gradual
fitness motion and deliberately lossy.

\begin{proposition}[fitness--rank reduction]\label{prop:fitness-rank}
Let $f_x:I\to\mathbb{R}$ be continuous for every candidate $x\in V$.
Suppose fitness values are distinct except at isolated times, every tie
involves exactly two candidates, and each crossing is transversal.  Then the
ranking induced by $(f_x(t))_{x\in V}$ is constant between crossings and
changes by one adjacent transposition at every crossing.
\end{proposition}

\begin{proof}
Continuity preserves every strict pairwise inequality on an interval without
a crossing.  At a crossing only two candidates, say $x$ and $y$, exchange
order.  If a third candidate lay strictly between them immediately before the
crossing, continuity would force another coincident crossing before $x$ and
$y$ could reverse.  Hence they are adjacent and their transposition is the
only ranking change.
\end{proof}

The reduction retains exactly what rank-based selection consumes: pairwise
order and crossing frequency.  It discards fitness gaps, gradients, basin
geometry, and the magnitude of improvement.  The theory therefore applies
directly to truncation, tournament, rank-based replacement, and
nondomination, but not to fitness-proportionate selection or adaptation rules
that consume numerical gaps.  The homogeneous Poisson/uniform-location law
is a tractable baseline for the crossing process, not a claim about every
dynamic landscape.

There are two cost interpretations.  With a comparison oracle, one probe is
one duel or preference query.  With scalar fitness, comparing two stale
candidates requires at most two current objective evaluations.  The full EA
experiments charge both evaluations and hold the total evaluation budget
fixed across maintenance policies.

\subsection{Abrupt environmental changes}\label{sec:shocks}

An abrupt change need not be represented by a fictitious gradual trajectory.
Let $\rk^-$ and $\rk^+$ be the true rankings immediately before and after a
change and call $J=K(\rk^-,\rk^+)$ its \emph{rank-shock size}.

\begin{proposition}[shock transfer]\label{prop:shock}
For every estimate $\est$,
\[
K(\est,\rk^+)\le K(\est,\rk^-)+J.
\]
After the shock, Theorems~\ref{thm:maxima} and~\ref{thm:topk} and
Proposition~\ref{prop:tournament} apply with the post-change distance.
\end{proposition}

\begin{proof}
The inequality is the triangle inequality for Kendall distance.  The
downstream statements are deterministic functions of the new true and
estimated rankings.
\end{proof}

Gradual drift thus supplies probabilistic age certificates, whereas abrupt
changes receive a deterministic damage certificate.  When a change detector
fires, pre-change age certificates are invalidated until the affected items
are probed again; continuing to advertise them would be unsound.

\section{A lower bound for oblivious probing}\label{sec:lower}

A schedule with a unit budget touches one pair per step while drift
disturbs $\alpha$ locations per step in expectation.  The next theorem
makes the resulting deficit quantitative for every oblivious schedule.  The
proof inspects a window of $w$ steps, Poissonizes the drift so that the
event counts $H_1,\dots,H_{n-1}$ at the $n-1$ locations become independent
$\mathrm{Poisson}(\mu)$ variables with $\mu = \alpha w/(n-1)$, and then
shows that for any location whose neighbors stayed quiet and whose resident
pair was never probed inside the window, the \emph{parity} of $H_\ell$,
the one quantity that decides the current order of the resident pair, is
independent of the entire transcript.  An estimator may be arbitrarily
intelligent; on those locations it is guessing a fair-ish coin it has
never observed.

\begin{theorem}[oblivious lower bound]\label{thm:lower}
Fix $\alpha > 0$ and an oblivious schedule with an arbitrary (adaptive,
randomized) estimator.  For every integer $1 \le w \le \min(t,\,n-2)$, with
$\mu = \alpha w/(n-1)$,
\[
\E\,K_t \;\ge\; (n-1-w)\;e^{-2\mu}\,\frac{1-e^{-2\mu}}{2}.
\]
\end{theorem}

\begin{proof}
Schedules drawn at random are handled by conditioning on the realization,
which is independent of the drift; so let the schedule be a fixed sequence
of item pairs.  Fix $t$ and consider the window $W$ of the last $w$ steps.
Write $\pi_0$ for the hidden ranking at the start of the window and
$Q_\ell = \{\pi_0^{-1}(\ell),\, \pi_0^{-1}(\ell+1)\}$ for the pair resident
at location $\ell$ then.  Let $P_W$ be the set of pairs probed during $W$;
$|P_W| \le w$.

\emph{Poissonization.}  The window contains
$N \sim \mathrm{Poisson}(\alpha w)$ drift events with i.i.d.\ uniform
locations; by the marking theorem the per-location counts
$(H_\ell)_{\ell=1}^{n-1}$ are independent $\mathrm{Poisson}(\mu)$.

For each location $\ell$ define the quiet-neighborhood event
\[
A_\ell \;=\; \{\,H_{\ell'} = 0 \text{ for } \ell' \in
\{\ell-1,\ell+1\} \cap \{1,\dots,n-1\}\,\},
\qquad \Prob(A_\ell) \ge e^{-2\mu}.
\]

\emph{Claim 1 (frozen residency).}  On $A_\ell$, the set of items occupying
ranks $\{\ell,\ell+1\}$ equals $Q_\ell$ throughout the window, their
relative order at time $t$ equals their order under $\pi_0$ flipped if and
only if $H_\ell$ is odd, and no item outside $Q_\ell$ has its rank
trajectory affected by the events at $\ell$.  Indeed, an item enters or
leaves the slot set $\{\ell,\ell+1\}$ only through an event at location
$\ell-1$ or $\ell+1$, which $A_\ell$ forbids; events at $\ell$ swap the two
residents in place and touch nobody else.

\emph{Claim 2 (transcript decoupling).}  Let $\mathcal{G}_\ell$ be the
$\sigma$-algebra generated by the entire history up to the start of the
window, the schedule, the estimator's coins, and the window's marked event
process restricted to locations other than $\ell$.  Then
$H_\ell$ is $\mathrm{Poisson}(\mu)$ independent of $\mathcal{G}_\ell$, and
$A_\ell \in \mathcal{G}_\ell$.  We claim that on
$A_\ell \cap \{Q_\ell \notin P_W\}$ every probe answer in the window is
$\mathcal{G}_\ell$-measurable.  A probe of a pair disjoint from $Q_\ell$
involves two items whose ranks are unaffected by events at $\ell$ (Claim 1)
and hence are determined by $\mathcal{G}_\ell$.  A probe of a pair
$\{x,z\}$ with $x \in Q_\ell$, $z \notin Q_\ell$ has answer
$[\rk(x) < \rk(z)] = [\ell < \rk(z)]$, because $\rk(x) \in \{\ell,\ell+1\}$
while $\rk(z) \notin \{\ell,\ell+1\}$; the answer is the same whichever
slot $x$ holds, hence again $\mathcal{G}_\ell$-measurable.  The direct probe
$\{x,y\} = Q_\ell$ is excluded by assumption.  Therefore the full
transcript $T$, and with it the estimate $\est_t = f(T,\text{coins})$, is
$\mathcal{G}_\ell$-measurable on this event, while the truth about the order
of $Q_\ell$ at time $t$ is $\pi_0$'s order flipped iff $H_\ell$ is odd: a
function of a parity independent of $\mathcal{G}_\ell$.

\emph{Per-location error.}  Conditioned on $\mathcal{G}_\ell$, on the event
$A_\ell \cap \{Q_\ell \notin P_W\}$ the estimator's ordering of $Q_\ell$ is
fixed, and it disagrees with the truth with probability at least
\[
\min\bigl(\Prob(H_\ell \text{ odd}),\, \Prob(H_\ell \text{ even})\bigr)
 = \frac{1 - e^{-2\mu}}{2},
\]
the parity probabilities of a Poisson($\mu$) count.

\emph{Counting.}  Distinct locations have distinct resident pairs (as
unordered pairs), and each misordered pair contributes one discordant pair
to $K_t$.  The map $\ell \mapsto Q_\ell$ is injective and $|P_W| \le w$, so
at most $w$ locations have $Q_\ell \in P_W$; note $\{Q_\ell \in P_W\}$ is
$\mathcal{G}_\ell$-measurable and independent of the window events.  Hence
\[
\E K_t \;\ge\; \sum_{\ell}
\E\Bigl[\mathbf{1}_{A_\ell}\,\mathbf{1}_{\{Q_\ell \notin P_W\}}\,
        \Prob\bigl(\est_t \text{ misorders } Q_\ell \,\big|\,
        \mathcal{G}_\ell\bigr)\Bigr]
\;\ge\; \frac{1-e^{-2\mu}}{2}\, e^{-2\mu} \,(n-1-w). \qedhere
\]
\end{proof}

\begin{corollary}[unit and higher rates]\label{cor:unit}
If $1 \le \alpha \le (n-1)/8$ and $n \ge 32$, then for all
$t \ge \lceil (n-1)/(4\alpha) \rceil$,
\(
\E K_t \ge (n-1)/12 .
\)
\end{corollary}

\begin{proof}
Take $w = \lceil (n-1)/(4\alpha)\rceil \le (n-1)/4 + 1 \le n-2$.  Then
$\mu = \alpha w/(n-1) \in [\tfrac14, \tfrac14 + \tfrac{\alpha}{n-1}]
\subseteq [\tfrac14, \tfrac38]$ since $\alpha \le (n-1)/8$.  The function
$g(\mu) = e^{-2\mu}(1-e^{-2\mu})/2$ is unimodal with maximum at
$\mu = \tfrac{\ln 2}{2} \approx 0.347$, so on $[\tfrac14,\tfrac38]$ it is
minimized at an endpoint: $g(\tfrac14) > 0.1193$ and
$g(\tfrac38) > 0.1246$.  Also
$n-1-w \ge \tfrac34(n-1) - 1 \ge 0.7177\,(n-1)$ for $n \ge 32$.  The
theorem gives $\E K_t \ge 0.7177 \cdot 0.1193\,(n-1) > (n-1)/12$.
\end{proof}

\begin{corollary}[low rates]\label{cor:low}
If $0 < \alpha \le 1$ and $n \ge 32$, then for all
$t \ge \lceil (n-1)/2 \rceil$,
\(
\E K_t \ge \alpha\,(n-1)/24 .
\)
\end{corollary}

\begin{proof}
Take $w = \lceil (n-1)/2 \rceil \le (n-1)/2 + 1$, so
$\mu \in [\tfrac{\alpha}{2},\, 0.5323\,\alpha]$ for $n \ge 32$.  Then
$e^{-2\mu} \ge e^{-1.0646\,\alpha} \ge e^{-1.0646} > 0.3448$, and by
concavity $1 - e^{-2\mu} \ge 1 - e^{-\alpha} \ge (1-e^{-1})\,\alpha
\ge 0.6321\,\alpha$ on $(0,1]$.  Finally
$n-1-w \ge \tfrac{n-3}{2} \ge 0.4677\,(n-1)$ for $n \ge 32$.  The theorem
gives $\E K_t \ge 0.4677 \cdot 0.3448 \cdot 0.3160\,\alpha(n-1)
> \alpha(n-1)/24$.
\end{proof}

\begin{remark}\label{rem:lb}
Together: for $n \ge 32$, any rate $\alpha \le (n-1)/8$, and all
$t \ge \lceil (n-1)/2\rceil$, every oblivious schedule satisfies
$\E K_t \ge \min(\alpha,1)\,(n-1)/24$.  Three readings.  (i) The bound is
linear in $\alpha$ below the unit rate and saturates above it; this is the
same shape the case study measures for the patrol's actual error, so the
scaling $\Theta(\min(\alpha,1)\,n)$ is the truth of the model, not an
artifact of either analysis.  (ii) Nothing was assumed about the
estimator: the bound is information-theoretic, paid in unobserved parities
rather than in computation.  (iii) The patrol itself is adaptive (its
targets depend on its own estimate), so the theorem does not bound it
directly; for adaptive probing at unit rate, a $\Theta(n)$ lower bound was
established by Anagnostopoulos, Kumar, Mahdian, and
Upfal~\cite{AnaKumMahUpf11} in the one-swap-per-step model.  The theorem
sharpens the oblivious case to explicit constants at every rate, the case
study finds the patrol operating a factor of about $13$ above the
$\tfrac{n-1}{24}$ line, and the entire remaining conversation is about
constants.
\end{remark}

\section{The comparison patrol}\label{sec:patrol}

\subsection{The structure}

The patrol index stores three arrays and a cursor: $S[1..n]$ with
$S[j]$ the item at estimated rank $j$ (so $S = \est^{-1}$), the inverse
$\mathrm{pos}[x] = \est(x)$, the timestamps $\mathrm{last}[x]$ of each
item's most recent probe, and a cursor $c \in \{1,\dots,n-1\}$.  Memory is
$3n + O(1)$ words.  One step costs one probe and $O(1)$ updates.

\begin{algorithm}[t]
\caption{One patrol step at time $t$, and the certified rank query.}
\label{alg:patrol}
\begin{algorithmic}[1]
\Procedure{PatrolStep}{$t$}
  \State $x \gets S[c]$;\quad $y \gets S[c+1]$
  \If{\textsc{Compare}$(x,y)$ answers ``$y$ precedes $x$''}
     \State swap $S[c] \leftrightarrow S[c+1]$;\quad
            $\mathrm{pos}[x] \gets c{+}1$;\quad $\mathrm{pos}[y] \gets c$
  \EndIf
  \State $\mathrm{last}[x] \gets t$;\quad $\mathrm{last}[y] \gets t$
  \State $c \gets c+1$; \textbf{if} $c = n$ \textbf{then} $c \gets 1$
  \Comment{boustrophedon variant: reverse direction at the ends}
\EndProcedure
\Statex
\Procedure{CertifiedRank}{$x, \delta, t$}
  \State $g \gets t - \mathrm{last}[x]$;\quad
         $D \gets D(g,\delta)$ \Comment{radius from
         Definition~\ref{def:cert}; $O(1)$ arithmetic}
  \State \Return $(\mathrm{pos}[x],\, D)$
  \Comment{proven part of the two-part contract:
           Definition~\ref{def:twopart}}
\EndProcedure
\Statex
\Procedure{CertifiedInterval}{$x, \delta, b, t$}
  \State $(p,\, D) \gets$ \Call{CertifiedRank}{$x, \delta, t$}
  \State \Return $[\,p - D - b,\; p + D + b\,] \cap [1, n]$
  \Comment{absolute interval under padding $b$:
           Definition~\ref{def:operational}}
\EndProcedure
\end{algorithmic}
\end{algorithm}

The cyclic patrol visits locations $1,2,\dots,n-1$ and wraps around; the
\emph{boustrophedon} patrol reverses direction at the ends, the turning
pattern of coverage paths in robotic planning~\cite{ChoPig98}.  Both repair
exactly the local discordances they meet, which makes a drift-free patrol a
bubble sort: with $\alpha = 0$ the estimate reaches the truth within $n-1$
cycles~\cite{Knuth98}.  Under drift, repair and decay balance, and the rest
of the paper quantifies the equilibrium.

Queries are constant-time reads: $\mathrm{rank}(x) = \mathrm{pos}[x]$,
$k$-th item $= S[k]$, and the order of any two items by comparing
positions.  Each answer is decorated by \textsc{CertifiedRank} with a
radius computed from the item's age alone.

\subsection{Ages}

\begin{proposition}[verification ages]\label{prop:age}
Let the cyclic patrol run with any drift.  Then:
\begin{enumerate}
\item[(i)] every item is probed at least once during every full cycle of
  $n-1$ consecutive steps, hence $a_t(x) \le 2(n-1)$ for every item $x$
  and every time $t$, deterministically;
\item[(ii)] for every schedule probing one pair per step there is, at every
  time $t \ge \lceil n/2\rceil - 1$, an item of age at least
  $\lceil n/2 \rceil - 1$;
\item[(iii)] consequently the patrol's worst-case age is within a factor
  $4 + o(1)$ of optimal.
\end{enumerate}
\end{proposition}

\begin{proof}
(i) An item's estimated position changes only when a probe involves it.
Fix a cycle and let $p$ be the position of item $x$ when the cycle's cursor
starts at location $1$.  As long as $x$ has not been probed in this cycle,
its position is still $p$, and the cursor reaches location
$\max(1, p-1)$ during the cycle; the probe there touches positions
$\{p-1, p\} \ni p$ (or $\{1,2\} \ni p$ if $p = 1$), so $x$ is probed.  Two
consecutive cycles span at most $2(n-1)$ steps between successive probes
of the same item.
(ii) A window of $L = \lceil n/2\rceil - 1$ steps touches at most
$2L \le n-1 < n$ items, so some item was untouched throughout and has age
at least $L$.
(iii) is the ratio of (i) and (ii).
\end{proof}

The boustrophedon patrol satisfies the same $O(n)$ age guarantee with
constant $4$ instead of $2$ in (i) (an item near a turning end is probed
twice in quick succession, then not for almost two traversals); the case
study finds the \emph{realized} maximum ages of both patrols
indistinguishable, concentrated at $\approx 1.0\,n$.

\subsection{Self-stabilization and recovery}\label{sec:stabilize}

Write $\ell_t(x) = \est_t(x) - \rk_t(x)$ for the \emph{overstatement} of
item $x$: positive when the board reports $x$ above its hidden rank.
Since $\sum_x \ell_t(x) = 0$, the positive parts carry half the footrule,
$\sum_x \ell_t(x)^+ = F_t/2$.  Call a cycle \emph{aligned} if the cursor
visits locations $1, 2, \dots, n-1$ in this order, and \emph{quiet} if no
drift event occurs during it.  One invariant of the quiet cycle carries
the rest of this section.

\begin{lemma}[bump]\label{lem:bump}
During a quiet aligned cycle: (i) every item whose overstatement at the
start of the cycle is at least $1$ is swapped down exactly once, so its
overstatement decreases by exactly one; (ii) every item whose
overstatement is at most $0$ at the start still has overstatement at most
$0$ at the end.  Consequently $\max_x \ell(x)^+$ decreases by exactly one
per quiet aligned cycle until it reaches zero.
\end{lemma}

\begin{proof}
Truth is frozen, so the cycle is one pass of a bubble sort with truthful
comparisons, and the standard winner invariant holds by induction: after
the probe at location $j$, position $j+1$ holds the item of largest
hidden rank among the items that occupied positions $1, \dots, j+1$ at
the start of the cycle.

(i) Let $x$ start at position $p = \est(x)$ with $\ell(x) \ge 1$.  Probes
at locations $1, \dots, p-2$ touch only positions up to $p-1$, so $x$ is
first involved at location $p-1$, as the right element.  The left element
is the winner of location $p-2$, the largest hidden rank among the
start-occupants of positions $1, \dots, p-1$.  Of those $p-1$ items, at
most $\rk(x) - 1 \le p - 2$ have hidden rank below $\rk(x)$, so at least
one ranks above $\rk(x)$, hence so does the winner, and the probe swaps
$x$ down to position $p-1$.  The cursor has now passed $x$: subsequent
probes touch positions $\ge p$, so $x$ moves exactly once.

(ii) An item rises only while it is the current winner.  When a winner
$y$ comes to rest at a position $q$ (it loses the probe at location $q$,
or $q = n$), the $q-1$ items then below it all belong to its start-prefix
and rank below it, so $\rk(y) \ge q = \est(y)$ and $\ell(y) \le 0$.
Items swapped down only decrease $\ell$, and untouched items keep it.

For the consequence: by (i) every item at the maximum positive
overstatement loses exactly one, and by (ii) no item at or below zero can
end above zero, so the maximum decreases by exactly one while positive.
\end{proof}

\begin{theorem}[exact self-stabilization]\label{thm:stabilize}
Run the patrol at $\alpha = 0$ from an arbitrary initial estimate with
$L = \max_x \ell_0(x)^+$.  If $L = 0$ the estimate already equals the
hidden order.  If $L \ge 1$, the patrol reaches the hidden order during
the $L$-th aligned cycle and not earlier: every item with
$\ell_0(x) = L$ is still misreported after $L-1$ complete cycles.  From
an arbitrary cursor phase, $(L+1)(n-1)$ probes always suffice.  The
reversed estimate attains $L = n-1$.
\end{theorem}

\begin{proof}
If $L = 0$ then all overstatements are $\le 0$ and sum to zero, so all
vanish and the estimate is the truth; this also proves the
$L$-cycle sufficiency, since after $L$ quiet aligned cycles the maximum
positive overstatement is zero by Lemma~\ref{lem:bump}.  For the lower
bound, an item with $\ell_0(x) = L$ has $\ell(x) \ge 1$ at the start of
each of the first $L-1$ cycles, is never a winner (it loses its first
probe of each cycle, as in the proof of the lemma), and therefore
descends exactly one position per cycle: after $L-1$ cycles its
overstatement is exactly $1 \ne 0$.  An arbitrary starting phase costs at
most one partial cycle, which the bump argument treats as a no-op:
$(L+1)(n-1)$ probes contain $L$ aligned cycles.  For the reversed
estimate, the item of hidden rank $1$ is reported at position $n$.
\end{proof}

The theorem is exact in practice, not only in statement: from the
reversed board the implementation sorts at measured sweep counts of
$62.016$, $126.008$, $254.004$, $510.002$, and $1022.001$ for $n = 64$
through $1024$, each equal to $L-1$ cycles plus a single probe, each
inside its predicted window $(L-1,\, L]$ (Section~\ref{sec:exp-ds}).
The sort completes at the first probe of the final cycle because the
last surviving overstatement has, by then, descended to the foot of the
board.  In the language of distributed protocols the
patrol is a \emph{self-stabilizing} structure in the sense of
Dijkstra~\cite{Dijkstra74}: started from arbitrary corruption it
converges to a legitimate state without external intervention, and here
the convergence time is not merely bounded but characterized to within
one cycle.

Self-stabilization prices abrupt environmental change.  An abrupt shock
(Proposition~\ref{prop:shock}) leaves the estimate intact and teleports
the truth; its damage to the board is exactly an overstatement profile.

\begin{theorem}[recovery and the rebuild crossover]\label{thm:recovery}
Let a shock change the hidden order so that the maximum overstatement
becomes $L$, and let no further drift occur during recovery.  Then:
\begin{enumerate}
\item[(i)] the patrol re-sorts within $(L+1)(n-1)$ probes;
\item[(ii)] a binary-insertion rebuild, which abandons the estimate and
  re-sorts from live comparisons, publishes the exact order after at most
  $C_{\mathrm{rs}}(n) = \sum_{m=2}^{n}\lceil \log_2 m\rceil
  \le n\lceil\log_2 n\rceil$ probes, whatever $L$ is;
\item[(iii)] the patrol bound is the smaller one exactly when
  $L + 1 < C_{\mathrm{rs}}(n)/(n-1)$.  At $n = 4096$,
  $C_{\mathrm{rs}} = 45{,}057 = 11.0\,(n-1)$ and the crossover sits at
  $L = 10 \approx \log_2 n$.
\end{enumerate}
\end{theorem}

\begin{proof}
(i) is Theorem~\ref{thm:stabilize} with an arbitrary phase.  (ii)
Inserting the $(m{-}1)$-st item into a sorted prefix of $m-1$ items by
binary search costs at most $\lceil\log_2 m\rceil$ truthful comparisons,
and with a frozen truth the completed sort is exact.  (iii) is
arithmetic: $\sum_{k=1}^{12} k\,2^{k-1} = 11 \cdot 2^{12} + 1$.
\end{proof}

The crossover converts the locality of a shock into a policy: a change
that overstates no candidate by more than about $\log_2 n$ positions is
cheapest to absorb in place, and a change that scrambles ranks globally
is cheapest to rebuild through.  A block reversal of width $w$ has
$L = w-1$; a handful of long-range exchanges already drives $L$ to
$\Theta(n)$.  Neither quantity must be known in advance, because the
structure can watch the one statistic it owns for free.

\begin{proposition}[the swap-count hybrid]\label{prop:hybrid}
A swap-free aligned quiet cycle certifies that the estimate equals the
hidden order.  Let the hybrid run the patrol, count swaps per cycle, stop
when a cycle is swap-free, and abandon the patrol for a binary-insertion
rebuild after $T_0 = \lceil C_{\mathrm{rs}}(n)/(n-1)\rceil$ swapping
cycles.  Under the quiet aftermath of Theorem~\ref{thm:recovery}, the
hybrid spends at most $(L+2)(n-1)$ probes when $L + 1 \le T_0$ and at
most $T_0(n-1) + C_{\mathrm{rs}}(n)$ probes otherwise; in both cases at
most
\[
2\,\min\bigl((L+1)(n-1),\; C_{\mathrm{rs}}(n)\bigr) \;+\; 2(n-1),
\]
with no knowledge of $L$, no clock, and no oracle access to the error.
\end{proposition}

\begin{proof}
If a full quiet cycle performs no swap, every adjacent estimated pair is
concordant, so $\rk(S[1]) < \dots < \rk(S[n])$ and the estimate is the
truth; the detector is sound.  If $L+1 \le T_0$: the patrol sorts within
$L+1$ cycles (phase included), the next cycle is swap-free, and
$L+1 \le T_0$ also gives $(L+1)(n-1) \le C_{\mathrm{rs}} + (n-1)$, so the
total $(L+2)(n-1)$ is at most $\min + 2(n-1)$ whichever term attains the
minimum.  Otherwise the hybrid pays $T_0(n-1) \le C_{\mathrm{rs}} +
(n-1)$ for the failed patrol phase plus one rebuild, while the minimum is
$C_{\mathrm{rs}}$ itself, since $(L+1)(n-1) > T_0(n-1) \ge
C_{\mathrm{rs}}$; the total is at most $2\,C_{\mathrm{rs}} + (n-1)$.
\end{proof}

The factor-two structure is the rent-or-buy trade of competitive
analysis~\cite{KarManRudSle88}: patrolling is renting, the rebuild is
buying, and the swap counter is the meter.  Section~\ref{sec:exp-shock}
runs the three policies against block reversals of width $4$ through $n$
and against nonlocal exchanges, at $\alpha = 0$ where the theorem speaks
and at $\alpha = 1$ where it does not; the measured recovery of the pure
patrol equals $L$ sweeps exactly at $\alpha = 0$, and the hybrid lands
within its bound in every run.

\subsection{The equilibrium, from both sides}\label{sec:equilibrium}

The patrol is a repair process: each probe removes at most one
discordance, and only an adjacent one.  This subsection bounds its
steady-state error from below by a theorem and from above by a
conjecture that the experiments audit at the level of individual
discordance lifetimes; the two sides agree on the order
$\Theta(\min(\alpha,1)\,n)$ and disagree only about a constant.

Call a maintainer \emph{location-oblivious} if the estimate location it
probes at step $t$ is a function of $t$ alone, possibly randomized
independently of the drift.  Both patrols qualify (their cursor ignores
the answers); the random adjacent baseline qualifies; repeated insertion
sort does not (its probe location depends on the comparison outcomes).
Theorem~\ref{thm:lower} does not bound the patrols, because their probed
\emph{pairs} depend on the evolving estimate; the next theorem closes
exactly that gap.

\begin{theorem}[equilibrium floor]\label{thm:floor}
Let a location-oblivious maintainer with consecutive probe footprint
(the cyclic or boustrophedon patrol) run against drift rate $\alpha >
0$.  For every $t \ge w$ and every $1 \le w \le n-2$,
\[
\E K_t \;\ge\;
\alpha w\,
\frac{1 - \dfrac{(3+2\alpha)\,w + 2}{n-1}}{1 + \dfrac{\alpha w}{n-1}}.
\]
\end{theorem}

\begin{proof}
Fix $t$ and the window $W$ of steps $\tau \in (t-w,\, t]$.  The cursor
positions are deterministic, so for each $\tau$ the set $T_\tau$ of
estimate positions touched by probes from step $\tau$ through step $t$
is a fixed interval with $|T_\tau| \le (t - \tau) + 2 \le w + 1$.

Consider a drift event $e$ of step $\tau$, at location $\ell$ drawn
uniformly from $\{1,\dots,n-1\}$ independently of the past, and let
$\{u,v\}$ be the resident pair at ranks $\{\ell, \ell+1\}$ just before
$e$.  Three failures can prevent $e$ from contributing a discordant pair
at time $t$.

\emph{Touched residents.}  If neither $u$ nor $v$ is ever probed in
$[\tau, t]$, their estimated positions, hence their estimated relative
order, are frozen.  A probe in the remaining window touches only items
standing at positions in $T_\tau$; each such position holds one item at
time $\tau$, and each item resides at no more than two locations (as
lower or upper resident), so at most $2(w+1)$ of the $n-1$ locations
have a resident that the window can touch:
$\Prob(\text{touched}) \le 2(w+1)/(n-1)$.

\emph{Re-flips.}  The hidden relative order of the specific pair
$\{u,v\}$ changes only at a drift event whose resident pair is again
$\{u,v\}$, which requires the pair to be rank-adjacent and the uniform
location to hit its boundary: each later event does so with probability
at most $1/(n-1)$, and the expected number of drift events in the
remaining window is at most $\alpha w$, so
$\Prob(\text{re-flip}) \le \alpha w/(n-1)$.

\emph{Born dead.}  If $\{u,v\}$ was discordant just before $e$, the
event repairs rather than creates.  The location is uniform and
independent of the state, so this has probability
$\E d_{\tau}/(n-1) \le \E K_\tau/(n-1)$, where $d_\tau$ counts
discordant resident pairs.  One step changes $K$ by at most $1 +
N_\tau$ with $\E N_\tau = \alpha$, so
$\E K_\tau \le \E K_t + (1+\alpha)w$.

If none of the three failures occurs, $e$ flips a concordant pair whose
estimated order then never changes and whose hidden order never changes
again: the pair is discordant at time $t$.  Distinct surviving events
have distinct pairs (a second event on the same pair is exactly a
re-flip of the first), and each discordant pair contributes one to
$K_t$.  Summing the expected $\alpha$ events per step over the $w$ steps
and subtracting the failure probabilities,
\[
\E K_t \;\ge\; \alpha w \Bigl(1 -
\frac{2(w+1) + \alpha w + (1+\alpha)w + \E K_t}{n-1}\Bigr),
\]
and solving for $\E K_t$ gives the statement.
\end{proof}

\begin{corollary}\label{cor:floor}
For $n \ge 128$, every $0 < \alpha \le (n-1)/8$, and every $t \ge n$,
the patrol holds $\E K_t \ge \min(\alpha, 1)\,(n-1)/24$, the same order
and the same constant scale as the oblivious bound of
Remark~\ref{rem:lb}.
\end{corollary}

\begin{proof}
Optimize the window.  For $\alpha \le 1$ take
$w = \lceil (n-1)/(2(3+2\alpha)) \rceil$.  Rounding adds at most one to
$w$, so the spoilage fraction $((3+2\alpha)w+2)/(n-1)$ is at most
$\tfrac12 + 9/(n-1) \le 0.58$ for $n \ge 128$, while
$\alpha w/(n-1) \le \alpha/(2(3+2\alpha)) + \alpha/(n-1) \le 0.11$.
With $3 + 2\alpha \le 5$ the theorem gives
$\E K_t \ge \alpha(n-1)\cdot\tfrac{0.42}{10 \cdot 1.11} \ge
\alpha(n-1)/27$ analytically; evaluating the bound at the exact integer
optimizer of $w$ (a closed form in the ancillary package, checked by an
exact test over the stated grid of $n$ and $\alpha$) raises every case
above $\alpha(n-1)/24$.  For $1 \le \alpha \le (n-1)/8$ the optimized
bound is increasing in $\alpha$ and stays above $(n-1)/24$.
\end{proof}

The patrol therefore cannot beat the $\Theta(\min(\alpha,1)\,n)$ law it
is measured to obey: disorder of this order is not a deficiency of the
cursor but the price of a unit budget, paid by the oblivious class
(Theorem~\ref{thm:lower}) and by the location-oblivious class alike.
What the theorems leave open is the constant, and we state its value as
a conjecture rather than dress a measurement as a theorem.

\begin{conjecture}[balance]\label{conj:balance}
For the cyclic patrol the stationary mean satisfies
$K^\ast/(\alpha n) \to \tfrac12$ as $\alpha \to 0$ (after $n \to
\infty$), and $K^\ast \le 0.56\,\alpha n$ for all $\alpha \le 4$ and
$n \ge 256$.
\end{conjecture}

The conjecture is a stock-flow prediction: a discordance is born at a
uniformly random location, waits about $(n-1)/2$ steps for the cursor,
and is repaired on the first visit unless drift moves it first, so
Little's law sizes the standing stock at $K^\ast \approx \alpha
\cdot (n-1)/2$.  What is rigorous in that sentence is the accounting
identity, not the lifetime estimate:

\begin{proposition}[flow identities]\label{prop:flow}
At stationarity, (i) the patrol's swap rate per step equals
$\alpha(1 - 2\bar q)$, where $\bar q$ is the stationary probability that
a uniformly chosen location holds a discordant resident pair; (ii)
$K^\ast = \bar\lambda\, \overline W$, where $\bar\lambda$ is the birth
rate of discordant pairs per step and $\overline W$ their mean lifetime
(Little's law~\cite{Little61,Stidham74}).
\end{proposition}

\begin{proof}
(i) A drift event creates a discordance at a concordant location and
removes one at a discordant location; with the uniform independent
location the expected drift change of $K$ per step is
$\alpha(1-2\bar q)$.  A swap changes $K$ by exactly $-1$.  At
stationarity the expected change vanishes.  (ii) is the stationary
version of $L = \lambda W$ for the birth-death flow of discordant
pairs, which is exactly conserved here because a pair's discordance
toggles only at its own drift events and its own repair swaps.
\end{proof}

The identities make the conjecture falsifiable at a finer grain than a
single number: the experiments can measure births and lifetimes
\emph{exactly}, because the relative order of a fixed pair changes only
at drift events of that pair and probe swaps of that pair, so a ledger
keyed by pairs tracks every birth and death without approximation.
Section~\ref{sec:exp-little} runs this audit: the product
$\bar\lambda\,\overline W$ reproduces the measured $K^\ast$ to within
$0.2\%$ across two decades of $\alpha$ and a factor of four in $n$, the
mean lifetime at $\alpha = 2^{-4}$ is $0.525$ sweeps against the
conjectured $0.5$, and $97\%$ of discordances are repaired on the first
cursor visit.  As $\alpha$ grows the birth rate per unit drift falls
(events increasingly land on already-discordant locations) while
lifetimes stretch (repairs compete with re-drift), and the two effects
cancel to keep $K^\ast/(\alpha n)$ inside $[0.508,\, 0.554]$ over the
entire grid: the constant the conjecture asserts, observed from inside
the queue rather than from its surface.

\paragraph{What the patrol buys.}
The constants fix the paper's position in the evolving-data line exactly.
Repeated insertion sort is the equilibrium-error champion of that line,
and this paper does not contest the title: the patrol concedes
$2.5$--$8.6\%$ of raw Kendall error by construction, because its cursor
keeps walking where an error-greedy maintainer would linger.  What the
concession purchases is the maintained interface, and the interface is the
contribution.  It consists of a deterministic verification age of $2(n-1)$
for every item at every instant (Proposition~\ref{prop:age}), a certified
displacement interval attached to every reported rank
(Section~\ref{sec:cert}), a per-item comparison load capped at two per
cycle, exact self-stabilization and recovery laws
(Theorems~\ref{thm:stabilize} and~\ref{thm:recovery}), a provable floor
showing the residual error is not the cursor's fault
(Theorem~\ref{thm:floor}), and transfer theorems that carry these
guarantees into downstream consumers (Theorems~\ref{thm:maxima}
and~\ref{thm:topk}, Proposition~\ref{prop:elitism}).  None of this is
available from the insertion sort: its realized worst-case age runs $85\%$
higher, fluctuates with the pass structure, admits no deterministic bound
at all (Section~\ref{sec:exp-cert}), and falls outside the
location-oblivious class, so even the floor must be argued differently.
The resulting decision rule is sharp.  When only the aggregate error of
the board matters, run repeated insertion sort.  The moment any consumer
acts on individual entries (a selection step, a reported frontier, a
service-level promise on rank queries), the certificate is the
deliverable, and the patrol supplies it at a single-digit premium that
Section~\ref{sec:experiments} prices exactly; the comparison against
all four natural reference points is fixed in one place by
Table~\ref{tab:positioning} in Section~\ref{sec:related}.
\section{Displacement certificates}\label{sec:cert}

A structure that may be $n$ probes stale must say so.  The patrol's
promise is built from the only motion statistic the model grants: between
probes, an item's hidden rank performs a lazy $\pm 1$ walk driven by the
drift events at its two flanking locations.

\begin{lemma}[motion tail]\label{lem:tail}
Fix an item $x$, times $t_0 < t$ with gap $g = t - t_0$, and
$\delta \in (0,1)$.  Put $\lambda = \alpha g$,
\[
\bar m = \lambda + \tfrac{L_1}{3} + \sqrt{\tfrac{L_1^2}{9} + 2\lambda L_1},
\qquad
\bar v = \frac{2\bar m}{n-1},
\qquad
D = \Bigl\lceil \tfrac{L_2}{3} +
      \sqrt{\tfrac{L_2^2}{9} + 2 L_2 \bar v} \Bigr\rceil,
\]
with $L_1 = \ln\tfrac{2}{\delta}$ and $L_2 = \ln\tfrac{4}{\delta}$.
If $D + 2 \le \rk_{t_0}(x) \le n-1-D$, then
\[
\Prob\Bigl(\,\sup_{t_0 \le s \le t}
   \bigl|\rk_s(x) - \rk_{t_0}(x)\bigr| \ge D\Bigr) \;\le\; \delta .
\]
\end{lemma}

\begin{proof}
Let $N$ be the number of drift events in $(t_0, t]$, a
$\mathrm{Poisson}(\lambda)$ variable.  The Bernstein-type Poisson tail
(see, e.g., \cite{BouLugMas13}) gives
$\Prob(N \ge \bar m) \le \exp\bigl(-u^2/(2(\lambda + u/3))\bigr)
 = \delta/2$ for $u = \bar m - \lambda$ as chosen.

Index time by events and let $R_k$ be the rank of $x$ after the $k$-th
event of the window.  An event moves $x$ iff its location lies in
$\{R_{k-1}-1,\, R_{k-1}\} \cap \{1,\dots,n-1\}$, and while
$R_{k-1} \in [2, n-1]$ both flanking locations exist, so conditionally on
the past the increment is $+1$ with probability $\tfrac1{n-1}$, $-1$ with
probability $\tfrac1{n-1}$, else $0$: a martingale difference with
$|\Delta| \le 1$ and conditional variance $\tfrac{2}{n-1}$.  Stop the walk
when it first leaves the band $(\rk_{t_0}(x) - D,\, \rk_{t_0}(x) + D)$;
because the band sits inside $[2, n-1]$ by the interior hypothesis, the
stopped process is a martingale with predictable quadratic variation at
most $2k/(n-1) \le \bar v$ during the first $\bar m$ events.  Freedman's
inequality~\cite{Freedman75}, two-sided, bounds the probability that the
stopped walk reaches magnitude $D$ while its variation is at most
$\bar v$ by $2\exp\bigl(-D^2/(2(\bar v + D/3))\bigr) \le \delta/2$ by the
choice of $D$.  A displacement of $D$ within the window requires either
$N > \bar m$ or such an exit among the first $\bar m$ events; the union
bound finishes.
\end{proof}

\begin{corollary}[time-varying drift rate]\label{cor:variable-rate}
Suppose adjacent events follow a non-homogeneous Poisson process with
instantaneous rate $\alpha(s)$ and still choose locations uniformly.  In
Lemma~\ref{lem:tail}, replace $\lambda=\alpha g$ by the integrated intensity
\[
\Lambda(t_0,t)=\int_{t_0}^{t}\alpha(s)\,ds.
\]
The same radius and failure probability follow.  If only an upper envelope
$\alpha(s)\le\bar\alpha(s)$ is known, using
$\int\bar\alpha(s)ds$ remains sound.
\end{corollary}

\begin{proof}
The event count of a non-homogeneous Poisson process on the interval is
Poisson with mean $\Lambda(t_0,t)$.  Conditional on that count, the stopped
rank walk and its variation bound are unchanged; an upper intensity envelope
only stochastically enlarges the count.
\end{proof}

\begin{definition}[certificate]\label{def:cert}
At query time $t$, the certificate of item $x$ at level $\delta$ is the
pair $(\mathrm{pos}[x],\, D(g,\delta))$ with $g = t - \mathrm{last}[x]$
and $D(\cdot)$ the radius of Lemma~\ref{lem:tail}.  Its guarantee is the
event of the lemma: \emph{the hidden rank of $x$ at time $t$ lies within
distance $D$ of the hidden rank $x$ held when it was last probed}, with
probability at least $1-\delta$, for every item whose certified band lies
inside the interior.  Items whose band reaches a boundary are reported
with the band clipped to $[1,n]$.
\end{definition}

\begin{remark}[the contract is two-part by necessity]\label{rem:contract}
The certificate certifies \emph{motion}, and only motion: the hidden rank
of $x$ now lies within $D$ of the hidden rank $x$ held at its last
verification.  It does not certify the reported position, because two
distinct gaps separate $\mathrm{pos}[x]$ from $\rk_t(x)$: the
\emph{residual} $|\mathrm{pos}[x] - \rk_{t_0}(x)|$, the placement error
already present at verification time, and the \emph{motion}
$|\rk_t(x) - \rk_{t_0}(x)|$ accrued since.  Lemma~\ref{lem:tail} bounds
the second and is silent on the first.  This division of labor is forced,
not chosen: the motion is a property of the drift process alone, so its
bound holds under any maintainer, any schedule, and any query times
(Proposition~\ref{prop:cover}), while the residual is a property of the
maintainer and admits no distribution-free bound (an adversarial estimate
makes it $n-1$).  A probe verifies relative order against one neighbor;
it does not certify absolute placement, and the interface never pretends
otherwise.
\end{remark}

\begin{definition}[two-part contract]\label{def:twopart}
A certified rank query is served as the pair of a \emph{proven part} and
a \emph{calibrated part}:
\[
\underbrace{(\mathrm{pos}[x],\; D(g,\delta))}_{\text{motion: theorem}}
\qquad\text{and}\qquad
\underbrace{(b,\; \hat\varepsilon_b)}_{\text{residual: calibration}},
\]
where the motion radius carries the guarantee of
Definition~\ref{def:cert} at level $\delta$, and the residual profile
asserts that the residual at verification time exceeds $b$ with
frequency at most $\hat\varepsilon_b$, an estimate produced by a
declared calibration protocol rather than by the drift model.
\end{definition}

\begin{proposition}[composition]\label{prop:compose}
Fix $\delta$, $b$, an interior query, and let $\varepsilon_b$ be the
true probability that the residual at the item's last verification
exceeds $b$.  Then
\[
\Prob\bigl(\,|\mathrm{pos}[x] - \rk_t(x)| > D(g,\delta) + b\,\bigr)
\;\le\; \delta + \varepsilon_b .
\]
\end{proposition}

\begin{proof}
$|\mathrm{pos}[x] - \rk_t(x)| \le |\mathrm{pos}[x] - \rk_{t_0}(x)| +
|\rk_{t_0}(x) - \rk_t(x)|$.  The first term exceeds $b$ with probability
at most $\varepsilon_b$; on its complement the total exceeds $D + b$
only if the motion reaches $D$, which has probability at most $\delta$
by Lemma~\ref{lem:tail}.  A union bound finishes.
\end{proof}

\begin{proposition}[calibration validity]\label{prop:calibration}
Let the residual profile be estimated on $m$ independent calibration
runs of the maintained process, with $\hat\varepsilon_b$ the average of
the per-run exceedance rates, and let $\varepsilon_b$ be the common
expected exceedance rate.  For any $\gamma \in (0,1)$, with probability
at least $1-\gamma$ over the calibration,
\[
\varepsilon_b \;\le\; \hat\varepsilon_b +
\sqrt{\frac{\ln(1/\gamma)}{2m}} .
\]
\end{proposition}

\begin{proof}
Per-run exceedance rates are independent, identically distributed, and
confined to $[0,1]$; Hoeffding's inequality applies to their mean.  The
within-run dependence of residual samples is immaterial because the run,
not the sample, is the independence unit.
\end{proof}

The protocol is deliberately split-sample: calibrate $b$ on one set of
seeds, audit it on a disjoint set.  Section~\ref{sec:exp-cert} runs the
split at $n = 1024$, $\alpha = 1$ with fifteen calibration and fifteen
audit runs: the calibrated padding is $b = 4$ at target exceedance
$10^{-2}$, the held-out exceedance is $0.0046$, and the audited coverage
of the composed interval at $\delta = 0.05$ is $0.99994$.  The
finite-sample certificate of Proposition~\ref{prop:calibration} is
honest about its units: at $m = 15$ runs it guarantees only
$\varepsilon_b \le 0.32$, and tightens as $m^{-1/2}$; the rigor lives in
the motion part, the sharpness in the audit, and the contract never
confuses the two.

\begin{definition}[operational interval]\label{def:operational}
For a padding $b \ge 0$, the \emph{operational interval} of item $x$ at
level $\delta$ is
$[\mathrm{pos}[x] - D - b,\ \mathrm{pos}[x] + D + b] \cap [1,n]$ with
$D = D(g,\delta)$, served by \textsc{CertifiedInterval} in
Algorithm~\ref{alg:patrol}: the two-part contract of
Definition~\ref{def:twopart} collapsed to one interval, with the
composed failure bound $\delta + \varepsilon_b$ of
Proposition~\ref{prop:compose}.  The padding is the calibrated price of
absolute placement; $b = 4$ covers the patrol's $99$th residual
percentile.  The case study measures the coverage of the full interval
and finds it above nominal already at $b = 0$, because the slack in the
motion tail absorbs the patrol's entire residual
(Table~\ref{tab:coverage}); the padded levels are for consumers who
refuse to spend that slack twice.
\end{definition}

\begin{proposition}[coverage]\label{prop:cover}
Under any maintainer and any drift rate, every interior certificate of
level $\delta$ is sound: the certified event fails with probability at
most $\delta$.  Soundness is a property of the drift process alone and
holds irrespective of the schedule, the estimator, and the query times.
\end{proposition}

\begin{proof}
Immediate from Lemma~\ref{lem:tail}: the certified event concerns only the
hidden walk between $t_0 = \mathrm{last}[x]$ and $t$, and the lemma's bound
is uniform over $t_0$, $t$, and everything the algorithm did.  (The probe
at $t_0$ is measurable with respect to the past, and the lemma's window
statistics are independent of the past.)
\end{proof}

\begin{remark}[the width does not grow with $n$]\label{rem:width}
For a relative age $g = \beta n$ the variation parameter is
$\bar v = 2\bar m/(n-1) \to 2\alpha\beta$ as $n$ grows, so the radius
tends to a constant depending only on $(\alpha\beta, \delta)$.  Under the
cyclic patrol $g \le 2(n-1)$ always (Proposition~\ref{prop:age}), so
\emph{every} certificate of a patrol-maintained board at $\alpha = 1$,
$\delta = 0.05$ has half-width at most $8$, at $n = 1024$ and at
$n = 4096$ alike.  A schedule without an age guarantee surrenders this:
the random-probe baseline of Section~\ref{sec:experiments} reaches
realized maximum ages near $0.55\, n\ln n$, pushing its worst certified
widths to $\pm 11$ and growing.  The certificate plays the role that a
certificate plays in kinetic data structures~\cite{BasGuiHer99}, with the
flight plans replaced by the drift statistics: the structure cannot see
trajectories, so it certifies a concentration property of the motion
instead of a deterministic one.
\end{remark}

\section{The evolving frontier}\label{sec:maxima}

The frontier layer is a transfer application of the maintained order,
the first of two consumers (the second is selection,
Section~\ref{sec:selection}), and its claims are scoped accordingly: we
do not propose a new dynamic-geometry framework, we prove that the
patrol's promises survive composition into a geometric report.  The
model boundary deserves to be stated as sharply as the theorems that
live inside it.

\begin{remark}[the maintenance model, against its three neighbors]
\label{rem:geo-model}
Kinetic maxima~\cite{BasGuiHer99} receive motion plans and process an
event queue of certificate failures; correctness is exact between
events.  Dynamic maxima~\cite{OveLee81} receive an explicit stream of
insertions and deletions; the structure sees every change.  Streaming
and sliding-window skylines~\cite{TaoPap06} see every arriving point
and pay for expiry.  Here the structure sees \emph{nothing}: no events,
no trajectories, no update stream, only two hidden orders drifting
behind a unit comparison budget, with all knowledge of motion reduced
to its statistics.  Exact maintenance is information-theoretically
impossible in this regime (Theorem~\ref{thm:lower} already forbids it
on one order), so the deliverable is necessarily different in kind: a
charging theorem from certified rank disorder to frontier error, in
global and frontier-local form, with the certificates carrying the
staleness.  The reduction is therefore the contribution, not a
disguise: two inversion budgets, plus the promises attached to them,
are \emph{all} the geometric layer is permitted to consume, and
Proposition~\ref{prop:local-necessity} shows the localized budget is
the right one among them.
\end{remark}

A finite planar point set in general position is determined, up to a
monotone reparametrization of each axis, by its pair of coordinate total
orders.  Let both orders drift independently at rate $\alpha$, each probed
by its own patrol (two probes per round in total).  The \emph{maxima}
$M$, the points dominated by no other in both coordinates, form the staircase
frontier of Kung, Luccio, and Preparata~\cite{KungLucPre75}; maintaining
maxima under explicit point updates is classical~\cite{OveLee81}, and here
the updates are hidden drift.  The reported frontier $\widehat M$ is the
set of maxima of the \emph{estimated} orders, assembled in $O(n)$ by one
sweep down the $x$-patrol's array $S^x$ keeping the running best estimated
$y$-position.  Reported and true frontier are tied together instance-wise:

\begin{theorem}[frontier charging]\label{thm:maxima}
Let $(\rk^x, \rk^y)$ and $(\est^x, \est^y)$ be any two pairs of rankings of
$V$, let $M$ and $\widehat M$ be the maxima of the true and the estimated
pairs, and let $K^x = K(\est^x, \rk^x)$, $K^y = K(\est^y, \rk^y)$.  Then,
deterministically,
\[
\bigl|M \diff \widehat M\bigr| \;\le\; 2\,(K^x + K^y).
\]
Moreover, for every $k \le \lfloor n/2 \rfloor$ there are instances with
$K^x + K^y = k$ and $|M \diff \widehat M| = k$, so the constant $2$ cannot
be improved below $1$.
\end{theorem}

\begin{proof}
Call a pair $\{p,q\}$ \emph{discordant} if its order flips between truth
and estimate in the $x$-ranking or in the $y$-ranking; there are at most
$K^x + K^y$ discordant pairs (a pair discordant in both coordinates is
counted twice by the right-hand side, which only helps).  We charge every
misclassified point to a discordant pair containing it.

Let $p \in M \setminus \widehat M$.  Since $p$ is not reported, some
$\hat q$ dominates $p$ in the estimated orders.  Since $p$ is truly
maximal, $\hat q$ does not dominate $p$ in the true orders, so at least
one coordinate order of the pair $\{p, \hat q\}$ flips: the pair is
discordant.  Charge $p$ to $\{p,\hat q\}$ (choosing, say, the
lexicographically least such witness).

Let $p \in \widehat M \setminus M$.  Some $q$ dominates $p$ truly; since
$p$ is reported maximal, $q$ does not dominate $p$ in the estimated
orders, and again $\{p,q\}$ is discordant; charge $p$ to it.

Every charged pair contains the point that charged it, so a fixed pair
absorbs at most two charges (one per endpoint).  Hence
$|M \diff \widehat M| \le 2 \cdot \#\{\text{discordant pairs}\}
 \le 2(K^x + K^y)$.

For tightness, place $p_1,\dots,p_n$ on the antidiagonal:
$\rk^x(p_i) = i$, $\rk^y(p_i) = n+1-i$, so every point is maximal.  Let the
estimate swap the $x$-order of $k$ disjoint adjacent pairs
$\{p_i, p_{i+1}\}$ and otherwise agree with the truth.  Each swap makes
$p_i$ dominate $p_{i+1}$ in the estimated orders ($\est^x(p_i) = i+1 >
\est^x(p_{i+1})$ and $\rk^y(p_i) > \rk^y(p_{i+1})$), deleting exactly
$p_{i+1}$ from $\widehat M$ and touching no other relation.  Thus
$K^x + K^y = k$ and $|M \diff \widehat M| = k$.
\end{proof}

\begin{corollary}\label{cor:maxima}
Two patrols at rate $\alpha$ maintain the frontier of a drifting planar
point set with
$\E\,|M_t \diff \widehat M_t| \le 2\,(\E K^x_t + \E K^y_t)$,
which at the measured equilibrium is at most $\approx 2.2\,\alpha n$.
\end{corollary}

The global bound spends one unit of frontier error per discordance
anywhere on the board.  Rereading its own proof shows that almost all of
that budget can never be used: every witness pair the charging argument
touches has \emph{both} endpoints on a frontier.

\begin{theorem}[frontier charging, localized]\label{thm:maxima-local}
In the setting of Theorem~\ref{thm:maxima}, let
$K_{\mathrm{loc}}$ be the number of pairs $\{p,q\} \subseteq M \cup
\widehat M$ whose order flips between truth and estimate in the
$x$-ranking or in the $y$-ranking.  Then, deterministically,
\[
\bigl|M \diff \widehat M\bigr| \;\le\; 2\,K_{\mathrm{loc}}
\;\le\; 2\,(K^x + K^y).
\]
\end{theorem}

\begin{proof}
We re-run the charging argument of Theorem~\ref{thm:maxima} with
maximal witnesses.  Let $p \in M \setminus \widehat M$.  Among the
points that dominate $p$ in the estimated orders, pick $\hat q$ maximal
with respect to estimated domination: if some $z$ estimated-dominated
$\hat q$ it would, by transitivity, also dominate $p$, contradicting the
maximality of $\hat q$ among $p$'s dominators; hence $\hat q \in
\widehat M$.  Since $p$ is truly maximal, $\hat q$ does not dominate $p$
in the true orders, so the pair $\{p, \hat q\}$ flips in at least one
coordinate, and both its endpoints lie in $M \cup \widehat M$.
Symmetrically, $p \in \widehat M \setminus M$ is charged to a pair
$\{p, q\}$ with $q \in M$ truly dominating $p$ and not
estimated-dominating it.  Every charged pair contains the point that
charged it, so a fixed pair absorbs at most two charges, giving the
first inequality.  The second holds because each pair counted by
$K_{\mathrm{loc}}$ is discordant in $x$ or in $y$, and there are at most
$K^x$ of the former and $K^y$ of the latter.
\end{proof}

Two constructions pin the refinement in place: it cannot be improved by
more than its factor two, and it cannot be replaced by any quantity that
looks at only one of the two frontiers.

\begin{proposition}[the localized budget is necessary]
\label{prop:local-necessity}
(i) For every $k \le \lfloor n/2 \rfloor$ there are instances with
$K_{\mathrm{loc}} = k$ and $|M \diff \widehat M| = k$, so
Theorem~\ref{thm:maxima-local} is tight up to its factor $2$.
(ii) Let $K_M$ (respectively $K_{\widehat M}$) count discordant pairs
with both endpoints in $M$ (respectively in $\widehat M$).  No bound of
the form $|M \diff \widehat M| \le f(K_M)$ or
$|M \diff \widehat M| \le f(K_{\widehat M})$ holds for any function
$f$ with $f(0) < \infty$: for every $k$ there are instances with
$K_M = 0$ and $|M \diff \widehat M| = k$, and instances with
$K_{\widehat M} = 0$ and $|M \diff \widehat M| = k$.
\end{proposition}

\begin{proof}
(i) is the tightness family of Theorem~\ref{thm:maxima}: on the
antidiagonal instance every point is truly maximal, the estimate swaps
$k$ disjoint $x$-adjacent pairs, the only discordant pairs are those
$k$ swapped pairs, both their endpoints lie in $M = V \subseteq M \cup
\widehat M$, and exactly $k$ points leave the reported frontier.

(ii) Take $n = 2k$ points $p_1, \dots, p_k$ and $q_1, \dots, q_k$ with
$x$-ranks $\rk^x(q_i) = 2i-1$, $\rk^x(p_i) = 2i$ and $y$-ranks
$\rk^y(p_i) = 2(k-i)+2$, $\rk^y(q_i) = 2(k-i)+1$.  Every $p_i$ is
maximal; each $q_i$ is dominated by $p_i$ and by nothing else (any
other point has larger $x$ exactly when it has smaller $y$), so
$M = \{p_1,\dots,p_k\}$.  Let the estimate swap the $x$-order of each
pair $\{q_i, p_i\}$ and agree with the truth otherwise.  In the
estimated orders $p_i$ no longer dominates $q_i$, and nothing else
ever did, so every $q_i$ enters the reported frontier:
$\widehat M = V$ and $|M \diff \widehat M| = k$.  The discordant pairs
are exactly the $k$ pairs $\{q_i, p_i\}$, each with one endpoint in
$M$ and one outside it, so $K_M = 0$ while the error is $k$.
Exchanging the roles of truth and estimate in the same instance gives
$K_{\widehat M} = 0$ with the same error, since Kendall disorder is
symmetric and the frontiers swap names.
\end{proof}

Internal disorder of the true frontier alone, or of the reported
frontier alone, certifies nothing; the damage is carried by cross pairs
between the two, which is exactly what $K_{\mathrm{loc}}$ counts and
what its evaluation over $M \cup \widehat M$ pays for.  The refinement
is geometry-sensitive through the single quantity
$|M \cup \widehat M|$: a sparse staircase exposes few pairs to the
charge, an antidiagonal one exposes all of them, and the worst case of
Theorem~\ref{thm:maxima} is recovered exactly there, so nothing is lost.
What is gained is measured in Section~\ref{sec:exp-maxima}: on
independent random orders the realized frontier error is $0.49$ per
snapshot, the localized bound evaluates to $1.1$, and the global bound
to $2.3 \times 10^3$; across copula-correlated geometries through the
fully antidiagonal family the localized bound tracks the realized error
within a factor of $2.3$, while the global certificate stands three to
four orders of magnitude above it.  The bound is still instance-wise and
distribution-free; it has simply learned to look at the staircase before
spending the Kendall budget.  In the language of multi-objective
optimization the staircase is the Pareto frontier, the nondominated set
of two drifting objectives; Section~\ref{sec:selection} completes that
reading.

\section{Selection under drifting fitness}\label{sec:selection}

The model of Section~\ref{sec:model} now returns to the reading that
motivates the paper, with nothing changed: a comparison-limited
evolutionary selection layer.  Items are
candidate solutions, the hidden order is the current fitness ranking,
adjacent transpositions are the generic ranking-level trace of gradual
environmental drift (Section~\ref{sec:fitness-model}), and the patrol
maintains a certified ranking from which selection decisions can be made
with explicit staleness guarantees.  This section makes the reading
formal.  The results are maintenance-layer guarantees: they bound what any
selection rule that consumes ranks can lose to drift, beneath and
independently of whatever search dynamics run above.  Three mechanisms
consume the estimate, and each acquires its own theorem: truncation
selection keeps the reported top $k$, tournament selection advances the
reportedly fitter of a random pair, and elitism protects reported
champions.

\subsection{Truncation selection transfers at the square root}

For a ranking $\sigma$ of $V$ and $1 \le k < n$ write
$T_k(\sigma) = \{x : \sigma(x) > n - k\}$ for its top-$k$ set, so that
truncation selection acting on the estimate keeps $T_k(\est_t)$ while the
environment rewards $T_k(\rk_t)$.  The damage is the symmetric
difference, and it transfers from the Kendall distance at the square
root, not linearly.

\begin{theorem}[selection transfer]\label{thm:topk}
For any rankings $\sigma, \tau$ of $V$, any $1 \le k < n$, and
$K = K(\sigma,\tau)$,
\[
\bigl|T_k(\sigma) \diff T_k(\tau)\bigr|
\;\le\; 2\bigl\lfloor \sqrt{K} \bigr\rfloor .
\]
Moreover, for every $m \le \min(k,\, n-k)$ there are pairs with $K = m^2$
and $|T_k(\sigma) \diff T_k(\tau)| = 2m$: the bound is attained at every
value of its right-hand side.
\end{theorem}

\begin{proof}
Let $A = T_k(\sigma) \setminus T_k(\tau)$ and
$B = T_k(\tau) \setminus T_k(\sigma)$; both top-$k$ sets have size $k$,
so $|A| = |B| =: m$.  Fix $x \in A$ and $y \in B$.  Then
$\sigma(x) > n-k \ge \sigma(y)$ while $\tau(y) > n-k \ge \tau(x)$, so the
pair $\{x,y\}$ is discordant.  The $m^2$ pairs of $A \times B$ are
pairwise distinct, whence $m^2 \le K$ and, $m$ being an integer,
$|T_k(\sigma) \diff T_k(\tau)| = 2m \le 2\lfloor\sqrt{K}\rfloor$.

For attainment let $\sigma$ be arbitrary and let $\tau$ exchange the two
blocks of $m$ items flanking the selection boundary: the items at
$\sigma$-ranks $n-k-m+1,\dots,n-k$ move up by $m$ and the items at ranks
$n-k+1,\dots,n-k+m$ move down by $m$, each block keeping its internal
order.  A pair flips precisely when it straddles the two blocks, so
$K = m^2$, while all $2m$ block members change sides of the boundary.
\end{proof}

\begin{remark}[footrule form]\label{rem:topk-footrule}
The same sets charged against displacement give
$F(\sigma,\tau) \ge m(m+1)$: ordering $A$ by estimated rank from the
boundary downward, its $i$-th member is reported at least $i$ positions
below a boundary it truly clears, and $B$ contributes symmetrically.
Either disorder functional certifies the square-root law; the Kendall
form is the one the maintainers of this paper control.
\end{remark}

\begin{corollary}[selection under the patrol]\label{cor:topk-patrol}
Under the cyclic patrol at equilibrium,
\[
\E\,\bigl|T_k(\rk_t) \diff T_k(\est_t)\bigr|
\;\le\; 2\,\E\sqrt{K_t} \;\le\; 2\sqrt{\E K_t}
\;\approx\; 2\sqrt{0.55\,\alpha n}
\]
simultaneously for every $k$.
\end{corollary}

The bound is global and the realized error is boundary-local: only
discordant pairs that straddle rank $n-k$ move items across the selection
boundary, and at equilibrium the patrol's displacements have $n$-free
magnitude (Table~\ref{tab:ages}), so the realized symmetric difference is
$O(1)$ per boundary.  The case study measures $0.5$--$1.2$ misplaced
candidates across every boundary from $k = 4$ to $k = 512$ against a
transfer bound of $46$ (Section~\ref{sec:exp-selection}); the
block-exchange family shows the gap belongs to the instance, not to the
theorem.

\subsection{Tournament selection: an exact identity}

\begin{proposition}[tournament error]\label{prop:tournament}
Let a binary tournament draw a pair $\{x,y\}$ uniformly at random and
advance the item the estimate ranks higher.  The decision contradicts the
hidden order with probability exactly $K_t \big/ \binom{n}{2}$.
\end{proposition}

\begin{proof}
The decision is wrong precisely when $\{x,y\}$ is discordant between
$\est_t$ and $\rk_t$, and exactly $K_t$ of the $\binom{n}{2}$ pairs are.
\end{proof}

The identity prices tournament selection on a maintained board with no
slack at all.  At patrol equilibrium the per-decision error rate is
$K^\ast/\binom{n}{2} \approx 1.1\,\min(\alpha,1)/(n-1)$, vanishing
linearly in the population size even though the absolute disorder
$K^\ast = \Theta(\alpha n)$ grows.  Some corruption is unavoidable in a
drifting environment, since Theorem~\ref{thm:lower} keeps
$K_t = \Omega(\min(\alpha,1)\,n)$ for the entire oblivious class; the
patrol holds the realized rate to within a measured factor $1.1$ of the
balance prediction and, unlike any error-only maintainer, certifies per
item how the corruption is distributed.  In the vocabulary of dueling
bandits~\cite{YueBroKleJoa12}, the maintainer is the arm-ranking layer,
and the identity converts its structural error into the per-duel decision
risk of whatever selection runs above it.

\subsection{Certified elitism}\label{sec:elitism}

Elitism is the selection decision that costs the most when wrong: a
candidate protected as a champion after it has silently fallen out of the
elite.  The displacement certificates of Section~\ref{sec:cert} turn
elitism into a per-decision guarantee.

\begin{proposition}[certified elitism]\label{prop:elitism}
Fix a level $\delta$, a padding $b \ge 0$, and $k$.  At query time $t$
let item $x$ have age $g$ and radius $D = D(g,\delta)$ with its certified
band interior, and suppose
$\mathrm{pos}[x] - D - b \ge n - k$.
If the residual at $x$'s last verification was at most $b$, then with
probability at least $1 - \delta$,
\[
\rk_t(x) \;>\; \mathrm{pos}[x] - D - b \;\ge\; n-k,
\qquad\text{that is,}\qquad x \in T_k(\rk_t).
\]
Unconditionally, the probability that the rule selects an item not truly
in the top $k$ is at most $\delta + \Prob(\mathrm{residual} > b)$.
\end{proposition}

\begin{proof}
A reported position changes only at probes involving the item, so
$\mathrm{pos}[x]$ at time $t$ equals its value at the last verification
$t_0$, and $|\mathrm{pos}[x] - \rk_{t_0}(x)| \le b$ by hypothesis.  On
the certificate event $|\rk_t(x) - \rk_{t_0}(x)| < D$
(Lemma~\ref{lem:tail}), and the two inequalities chain.
\end{proof}

Under the patrol the rule is uniform rather than adaptive: ages never
exceed $2(n-1)$, so $D \le 8$ at $\delta = 0.05$, $\alpha = 1$ for every
$n$ (Remark~\ref{rem:width}), and with the calibrated $b = 4$ the margin
is twelve positions.  \emph{Any candidate reported twelve or more
positions inside the boundary can be selected with a per-decision
guarantee, on a board of any size, at any time.}  No maintainer without
an age bound offers this: under the random-probe baseline the margin
tracks the certified width of the oldest item, which grows like
$\sqrt{\log n}$ and keeps growing.  The case study runs the rule at
$k \in \{16, 64, 256\}$, selects $559{,}800$ candidates, and records one
wrong pick (precision $0.999998$) while still certifying $96\%$ of the boundary's capacity at
$k = 256$ (Section~\ref{sec:exp-selection}).

\subsection{Selection reliability over time}\label{sec:reliability}

The identities above price single decisions; search runs on streams of
them, so the layer should also say how long a stream stays clean.  The
statement is short because the machinery already paid for it.

\begin{proposition}[certified failure time]\label{prop:failure-time}
Issue certified picks by the margin rule at level $\delta$ with padding
$b$, and let $\varepsilon_b$ be the residual exceedance probability.
Among any $N$ picks the expected number of wrong ones is at most
$N(\delta + \varepsilon_b)$, and the index $T$ of the first wrong pick
satisfies $\E\,T \ge 1/\bigl(2(\delta+\varepsilon_b)\bigr)$.
\end{proposition}

\begin{proof}
Each certified pick is wrong with probability at most $\delta +
\varepsilon_b$ (Propositions~\ref{prop:compose} and~\ref{prop:elitism},
unconditional form); linearity of expectation gives the count, with no
independence needed.  Writing $q = \delta + \varepsilon_b$, a union
bound gives $\Prob(T \le N) \le Nq$, so with $M = \lfloor 1/q \rfloor$,
\[
\E\,T \;=\; \sum_{N \ge 0} \Prob(T > N)
\;\ge\; \sum_{N=0}^{M} (1 - Nq)
\;=\; (M{+}1)\Bigl(1 - \frac{qM}{2}\Bigr)
\;\ge\; \frac{M+1}{2} \;\ge\; \frac{1}{2q}. \qedhere
\]
\end{proof}

While no certified failure has occurred, every champion the rule
protects is a true member of the top $k$: elitist tracking-loss events,
the silent protection of a fallen elite that costs dynamic optimizers
most, occur at rate at most $\delta + \varepsilon_b$ per certified
decision, deterministically accounted.  The measured slack is large in
the structure's favor: across $559{,}800$ certified picks the case study
records one error, a realized rate of $1.8 \times 10^{-6}$ against a
contractual $\delta + \varepsilon_b \approx 0.055$
(Section~\ref{sec:exp-selection}).

Per-decision error rates also propagate into whole-algorithm guarantees,
not by analogy but by coupling: any conclusion proved for an idealized
evolutionary algorithm that sees the true ranking transfers to the same
algorithm running on the maintained board, with an additive slack the
tournament identity prices exactly.

\begin{proposition}[tournament coupling]\label{prop:coupling}
Fix any algorithm whose only access to fitness is through binary
tournaments on independently, uniformly drawn pairs, judged on the
maintained board, and let its idealized twin judge the same tournaments
on the hidden order.  Couple the two runs on identical randomness
(pair draws, variation, replacement, tie-breaking).  Then the
trajectories coincide until the first tournament that draws a pair
discordant at that moment, and for every event $A$ determined by the
first $T$ tournaments,
\[
\bigl|\,\Prob_{\mathrm{maintained}}(A) - \Prob_{\mathrm{true}}(A)\,\bigr|
\;\le\; \sum_{t=1}^{T} \frac{\E\,K_{(t)}}{\binom{n}{2}},
\]
where $K_{(t)}$ is the board's Kendall distance when the $t$-th
tournament is drawn, along the coupled trajectory.  At a stationary
level $\bar K$ the index of the first diverging decision has
expectation at least $\binom{n}{2}/(2\bar K)$.
\end{proposition}

\begin{proof}
Under the coupling the two runs share every random input, and as long
as no drawn pair is discordant, every tournament returns the same
winner in both, so the full states (population, board, counters) remain
identical by induction.  The $t$-th drawn pair is uniform and
independent of the current state, hence discordant with probability
exactly $K_{(t)}/\binom{n}{2}$ conditionally on the past
(Proposition~\ref{prop:tournament}); a union bound over the first $T$
tournaments bounds the probability that the runs have diverged, which
dominates the difference in probability of any event they determine.
The expectation bound is the computation of
Proposition~\ref{prop:failure-time} with
$q = \bar K/\binom{n}{2}$.
\end{proof}

At the patrol equilibrium $\bar K \approx 0.55\min(\alpha,1)\,n$ the
horizon is $\binom{n}{2}/(2\bar K) \approx 0.45\,(n-1)/\min(\alpha,1)$
tournaments: for a population of $512$ at unit drift, about $230$
consecutive selection decisions are expected to be exactly the
decisions of the idealized algorithm.  In particular every hitting-time
or tracking statement available for the idealized process, for
instance in the dynamic runtime-analysis
literature~\cite{Droste03,RohLehYao09,LisWit16}, holds for the
maintained process up to the stated additive slack, with
$\E K_{(t)}$ supplied by the equilibrium theory of
Section~\ref{sec:equilibrium} rather than assumed.

These propositions also mark, deliberately, where the layer's authority
ends.  Hitting times, tracking error, and diversity dynamics couple the
selection stream to the landscape, the variation operators, and the
change process; the frequency-and-magnitude axis of dynamic runtime
analysis~\cite{RohLehYao09,LisWit16} parameterizes exactly that
coupling.  This paper contributes the interface those analyses consume:
exact per-decision error rates, failure-time floors, and a measurable
displacement statistic that locates a landscape relative to the
recovery crossover of Theorem~\ref{thm:recovery}.  When drift
decorrelates ranks faster than the budget re-verifies them, no ranking
layer can convert order into progress, and the certificates do not
pretend to: they report the obsolescence itself, through ages, withheld
post-shock coverage, and the swap-rate detector, which is precisely the
information an evaluation manager needs to trigger the responses of
Section~\ref{sec:exp-ea}.

\subsection{Two drifting objectives}

Theorem~\ref{thm:maxima} is the two-objective member of the same family.
The maxima of a planar point set under its coordinate orders are exactly
the Pareto-optimal candidates under two objectives: the nondominated set
that multi-objective evolutionary methods in the NSGA
line~\cite{DebPraAgaMey02} maintain as their population elite.  Read this
way, the frontier theorem states that if two objectives drift and each is
tracked by a patrol on its induced ranking, the reported nondominated set
differs from the true one by at most $2(K^x + K^y)$, instance-wise, with
both Kendall distances controlled by their patrols and the assembly
costing $O(n)$ per generation.  The measured frontier error of $0.48$
points per snapshot (Section~\ref{sec:exp-maxima}) is the two-objective
analogue of the boundary locality above: discordances corrupt selection
only where they touch the elite.

\section{Experimental evaluation}\label{sec:experiments}

The paper profile uses thirty paired seeds.  The maintenance grid is
$n\in\{256,1024,4096,16384\}$ and
$\alpha\in\{2^{-4},2^{-2},1,4,16\}$, with $80$ measured sweeps after burn-in;
an additional ten-seed run reaches $n=65{,}536$.  Frontier experiments
warm-start the estimates because they test maintenance rather than initial
construction.  The robustness and full-EA protocols are specified in their
subsections.  Every per-seed observation, confidence interval, statistical
test, source hash, and command line is included in \texttt{anc/results/}.

Four maintainers are compared throughout: the \emph{cyclic} and
\emph{boustrophedon} patrols, the \emph{repeated insertion sort} (the
maintainer with proven $\Theta(n)$ steady state at unit
rate~\cite{BesaDEGJ18}), and the \emph{random adjacent probe}, which
relaxes exactly one ingredient of the patrol, the deterministic cursor,
and prices it.  The selection study of
Section~\ref{sec:exp-selection} adds a fifth, the \emph{generational
re-sorter}: it re-sorts the whole population by binary insertion (about
$n \log_2 n$ live comparisons per generation) and serves the previously
published generation meanwhile, the ranking-level shadow of full periodic
re-evaluation under the identical unit budget.

\subsection{The structure as a data structure}\label{sec:exp-ds}

Table~\ref{tab:ds-microbench} measures the patrol as an artifact rather
than as a process: wall-clock step cost, memory footprint, the
deterministic age bound, and access locality, reported through the
\emph{access span}, the expected distance between consecutive probed
locations, a portable proxy for cache behavior that does not depend on
hardware counters.  The patrol's step is a constant-time, span-one
operation on three flat arrays; the random baseline performs the same
arithmetic yet jumps $\Theta(n)$ locations per probe and pays for it both
in time and, unboundedly, in age.  The last two rows close the loop on
two theorems at once.  The certificate row compares the certified radius
at age $n$ with the realized $99$th displacement percentile, the gap
being the price of Freedman generality (Section~\ref{sec:exp-cert}).
The stabilization row starts the patrol from the reversed board, the
worst case of Theorem~\ref{thm:stabilize}, and observes sorting at
$1022.001$ sweeps for $n = 1024$: one probe past $L-1$ full cycles,
inside the predicted window $(L-1,\, L]$, with the same exactness at
every size tested down to $n = 64$.

\begin{table}[H]
\centering\small
\caption{Data-structure microbenchmarks at $n = 65{,}536$ (ten seeds;
time per maintained step on one core).  Space counts live words; age is
the deterministic bound; span is the expected distance between
consecutive probed locations.  Below the rule: certificate tightness at
gap $g = n$, and the reversed-start stabilization measurement.}
\label{tab:ds-microbench}
\begin{tabular}{lcccc}
\toprule
structure & ns / step & space (words) & age bound & access span\\
\midrule
comparison patrol & $1160$ & $196{,}612$ & $131{,}070$ & $1.0$\\
repeated insertion & $1163$ & $196{,}612$ & unbounded & $1.0$\\
random adjacent & $1276$ & $229{,}380$ & unbounded & $21845.0$\\
\midrule
certificate at $g=n$ & -- & -- & $D=6$ & $p_{99}=4.0$\\
reversed start, zero drift & -- & -- & -- & $1022.0$ sweeps
\\
\bottomrule
\end{tabular}
\end{table}

\subsection{Equilibrium and scaling}\label{sec:exp-steady}

Table~\ref{tab:steady} and Figure~\ref{fig:steady} carry the equilibrium
story.  Four reads.  First, $K^\ast/n$ is constant in $n$ for every
maintainer: the $\Theta(\alpha n)$ equilibrium of
Remark~\ref{rem:lb} realized.  Second, the floor is now two-sided
theory rather than analogy: at $\alpha = 1$ the cyclic patrol measures
$K^\ast = 0.553\,n$ against its own proven floor of $0.045\,n$
(Theorem~\ref{thm:floor}, a factor $12.2$) and the oblivious bound
$(n-1)/12$ of Corollary~\ref{cor:unit}; every remaining disagreement
between theory and measurement lives in the constant of
Conjecture~\ref{conj:balance}, which Section~\ref{sec:exp-little} audits
directly.  Third, the collapse: dividing by
$\alpha n$ flattens every cyclic-patrol configuration onto
$0.465$--$0.556$.  Rates through $\alpha=4$ remain close to the conjectured
balance constant $\tfrac12$; at $\alpha=16$ the decline marks the onset of repair
saturation rather than a new scaling law.  Fourth, the price of patrol structure is small and
the price of no structure is not: the boustrophedon patrol sits $2.5\%$
above repeated insertion sort at $n = 4096$, the cyclic patrol $8.6\%$
above, while the random adjacent probe (identical except that its cursor
is memoryless) pays a factor $2.0$ in error and, as
Table~\ref{tab:ages} shows, an unbounded factor in age.

\begin{table}[H]
\centering\small
\caption{Steady-state Kendall distance per item, $K^\ast/n$, at
$\alpha = 1$ (thirty seeds, $80$ measured sweeps).  The flat columns show the
$\Theta(n)$ equilibrium; random probing pays roughly a factor two.}
\label{tab:steady}
\begin{tabular}{lcccc}
\toprule
maintainer & $n{=}256$ & $1024$ & $4096$ & $16384$\\
\midrule
cyclic patrol & $0.550$ & $0.552$ & $0.555$ & $0.555$\\
boustrophedon patrol & $0.520$ & $0.522$ & $0.524$ & $0.523$\\
repeated insertion & $0.505$ & $0.509$ & $0.511$ & $0.511$\\
random adjacent & $1.004$ & $1.023$ & $1.028$ & $1.026$
\\
\bottomrule
\end{tabular}
\end{table}

\begin{figure}[H]
\centering
\begin{tikzpicture}
\begin{axis}[width=0.49\textwidth, height=5.4cm,
  xmode=log, log basis x=2, xlabel={$n$}, ylabel={$K^\ast/n$},
  ymin=0, ymax=1.65, ytick={0,0.5,1,1.5}, xtick={256,1024,4096,16384},
  xticklabels={$256$,$1024$,$4096$,$16384$},
  legend pos=north west, legend cell align=left]
\addplot+[mark=*, mark size=1.6pt, error bars/.cd, y dir=both, y explicit]
  table[x=n, y=K, y error=sd]{figdata/steady_cyclic.dat};
\addlegendentry{cyclic patrol}
\addplot+[mark=square*, mark size=1.5pt, error bars/.cd, y dir=both,
  y explicit] table[x=n, y=K, y error=sd]{figdata/steady_boustrophedon.dat};
\addlegendentry{boustrophedon}
\addplot+[mark=triangle*, mark size=1.9pt, error bars/.cd, y dir=both,
  y explicit] table[x=n, y=K, y error=sd]{figdata/steady_insertion.dat};
\addlegendentry{repeated insertion}
\addplot+[mark=diamond*, mark size=1.9pt, error bars/.cd, y dir=both,
  y explicit] table[x=n, y=K, y error=sd]{figdata/steady_random.dat};
\addlegendentry{random adjacent}
\addplot[dashed, domain=256:16384] {1.0/12};
\addlegendentry{oblivious bound $\tfrac{1}{12}$}
\end{axis}
\end{tikzpicture}\hfill
\begin{tikzpicture}
\begin{axis}[width=0.49\textwidth, height=5.4cm,
  xmode=log, log basis x=2, xlabel={$\alpha$},
  ylabel={$K^\ast/(\alpha n)$},
  ymin=0, ymax=0.75,
  xtick={0.0625,0.25,1,4,16},
  xticklabels={$2^{-4}$,$2^{-2}$,$1$,$2^{2}$,$2^{4}$},
  legend pos=south east, legend cell align=left]
\addplot+[mark=*, mark size=1.6pt, error bars/.cd, y dir=both, y explicit]
  table[x=alpha, y=K, y error=sd]{figdata/collapse_256.dat};
\addlegendentry{$n = 256$}
\addplot+[mark=square*, mark size=1.5pt, error bars/.cd, y dir=both,
  y explicit] table[x=alpha, y=K, y error=sd]{figdata/collapse_4096.dat};
\addlegendentry{$n = 4096$}
\addplot[dashed, domain=0.05:18] {0.5};
\addlegendentry{conjectured balance $\tfrac12$}
\end{axis}
\end{tikzpicture}
\caption{Left: equilibrium error per item is flat in $n$ for all four
maintainers ($\alpha=1$); the dashed line is the oblivious lower bound of
Corollary~\ref{cor:unit}, and the patrol's own floor
(Theorem~\ref{thm:floor}) sits at $0.045$ on this axis.  Right: the
scaling collapse.  Plotting $K^\ast/(\alpha n)$ collapses twenty patrol
configurations across three decades of drift rate; the dashed line is
the constant of Conjecture~\ref{conj:balance}.}
\label{fig:steady}
\end{figure}
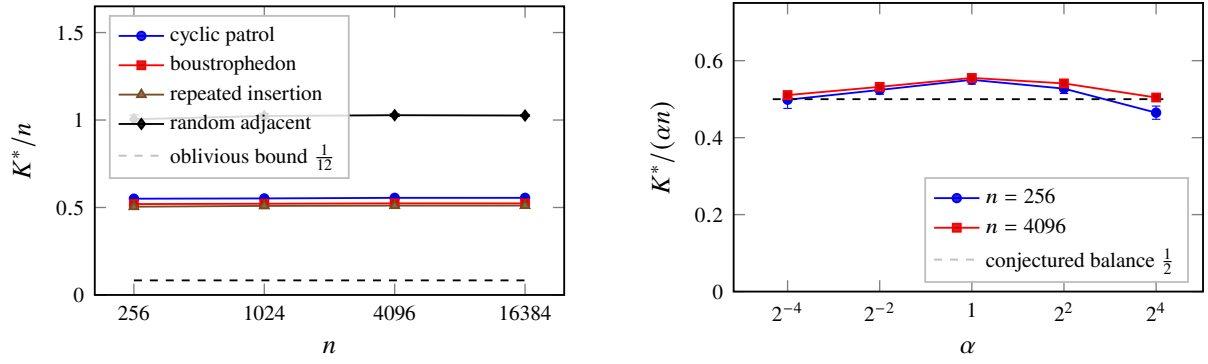

The footrule tells the same story in more humane units: under the cyclic
patrol $F^\ast/n \in [0.994, 1.007]$ for every $n$ tested.  \emph{A
patrol-maintained board at unit drift is off by one position per entry, on
average, at every scale.}  Against the Diaconis--Graham
inequalities~\eqref{eq:dg} the measured ratio $F^\ast/K^\ast \approx 1.81$
sits near the upper end $2$: equilibrium disorder is dominated by many
small displacements rather than few large ones, which is precisely the
regime certificates can afford to be tight in.

\subsection{Approach to equilibrium and unattended decay}

Figure~\ref{fig:traj} shows the two halves of the maintenance trade.
Unattended, the estimate loses one pair per event essentially exactly as
long as discordances are sparse (the envelope of Lemma~\ref{lem:decay} is
tight to $4\%$ at $t = n/4$), so every idle step has a known, linear price.
Attended, the price stops accruing: from a cold sorted start the patrol's
error climbs for about three sweeps and then holds the equilibrium of
Table~\ref{tab:steady} indefinitely.  Equilibrium is reached at the speed
of the budget, not of the mixing time.

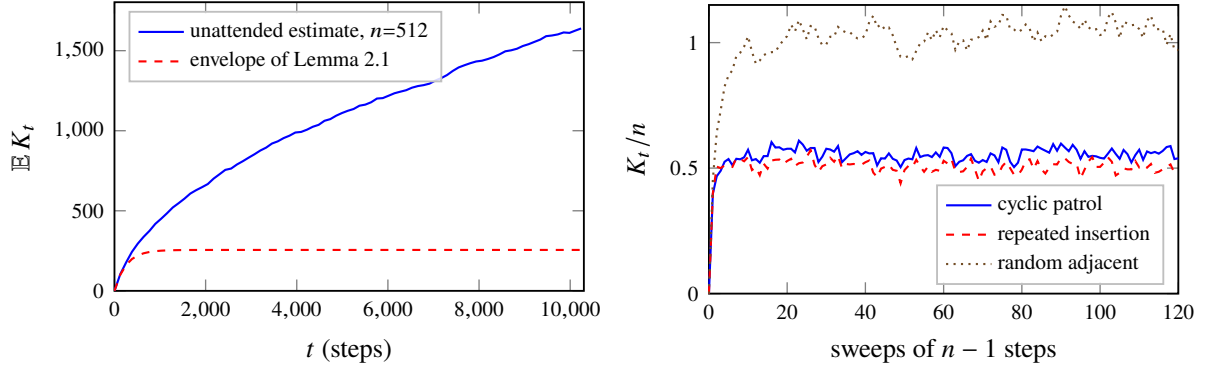
\begin{figure}[H]
\centering
\begin{tikzpicture}
\begin{axis}[width=0.49\textwidth, height=5.4cm,
  xlabel={$t$ (steps)}, ylabel={$\E\,K_t$},
  xmin=0, xmax=10300, ymin=0,
  legend pos=north west, legend cell align=left,
  scaled x ticks=false]
\addplot+[mark=none, thick] table[x=t, y=K]{figdata/decay.dat};
\addlegendentry{unattended estimate, $n{=}512$}
\addplot+[mark=none, dashed, thick] table[x=t, y=bound]{figdata/decay.dat};
\addlegendentry{envelope of Lemma~\ref{lem:decay}}
\end{axis}
\end{tikzpicture}\hfill
\begin{tikzpicture}
\begin{axis}[width=0.49\textwidth, height=5.4cm,
  xlabel={sweeps of $n-1$ steps}, ylabel={$K_t/n$},
  xmin=0, xmax=120, ymin=0, ymax=1.15,
  legend pos=south east, legend cell align=left]
\addplot+[mark=none, thick] table[x=sweep, y=K]{figdata/traj_cyclic.dat};
\addlegendentry{cyclic patrol}
\addplot+[mark=none, thick, dashed]
  table[x=sweep, y=K]{figdata/traj_insertion.dat};
\addlegendentry{repeated insertion}
\addplot+[mark=none, thick, dotted]
  table[x=sweep, y=K]{figdata/traj_random.dat};
\addlegendentry{random adjacent}
\end{axis}
\end{tikzpicture}
\caption{Left: with zero probes, disorder accrues at unit rate and the
exact envelope $\tfrac{n-1}{2}(1-e^{-2\alpha t/(n-1)})$ of
Lemma~\ref{lem:decay} tracks the sparse-disorder regime.  Right: patrol and
insertion reach equilibrium within three sweeps ($n=1024$, $\alpha=1$), while
random probing converges to a plateau roughly twice as high.}
\label{fig:traj}
\end{figure}

\subsection{The equilibrium audited from inside}\label{sec:exp-little}

The constant of Conjecture~\ref{conj:balance} is not audited by reading
$K^\ast$ off the board but by metering the queue that produces it.
Because the relative order of a fixed pair changes only at drift events
of that pair and at probe swaps of that pair, a ledger keyed by pairs
records every birth and every death of a discordance exactly; the ledger
size reproduces the Kendall distance at every checkpoint by construction,
and the experiment asserts this identity rather than assuming it.
Table~\ref{tab:lifetimes} reports the flow at $n = 4096$ over ten seeds
and eighty measured sweeps per configuration.

\begin{table}[H]
\centering\small
\caption{The discordance queue at $n = 4096$.  Births and lifetimes are
exact ledger counts; the Little column is the ratio
$\bar\lambda\,\overline W / K^\ast$ (Proposition~\ref{prop:flow});
first-pass is the fraction of discordances repaired on the first cursor
visit; repair share is the fraction of deaths due to the patrol rather
than to drift cancellation.}
\label{tab:lifetimes}
\begin{tabular}{lcccccc}
\toprule
$\alpha$ & $K^\ast\!/(\alpha n)$ & $\bar\lambda/\alpha$
  & $\overline W$ (sweeps) & Little & first-pass & repair share\\
\midrule
$2^{-4}$ & $0.508$ & $0.969$ & $0.525$ & $1.002$ & $0.970$ & $0.969$\\
$2^{-2}$ & $0.531$ & $0.889$ & $0.598$ & $1.001$ & $0.894$ & $0.875$\\
$1$ & $0.554$ & $0.719$ & $0.771$ & $1.000$ & $0.766$ & $0.607$\\
$2^{2}$ & $0.542$ & $0.582$ & $0.930$ & $1.000$ & $0.750$ & $0.283$
\\
\bottomrule
\end{tabular}
\end{table}

\begin{figure}[H]
\centering
\begin{tikzpicture}
\begin{axis}[width=0.7\textwidth, height=6.2cm,
  ymode=log, xlabel={lifetime (sweeps of $n-1$ steps)},
  ylabel={survival fraction},
  xmin=0, xmax=4, ymin=2e-5, ymax=1.3,
  legend pos=north east, legend cell align=left]
\addplot+[mark=none, thick]
  table[x=sweeps, y=survival]{figdata/life_tail_0p0625.dat};
\addlegendentry{$\alpha = 2^{-4}$}
\addplot+[mark=none, thick, dashed]
  table[x=sweeps, y=survival]{figdata/life_tail_0p25.dat};
\addlegendentry{$\alpha = 2^{-2}$}
\addplot+[mark=none, thick, dotted]
  table[x=sweeps, y=survival]{figdata/life_tail_1p0.dat};
\addlegendentry{$\alpha = 1$}
\addplot+[mark=none, thick, dashdotted]
  table[x=sweeps, y=survival]{figdata/life_tail_4p0.dat};
\addlegendentry{$\alpha = 4$}
\end{axis}
\end{tikzpicture}
\caption{Lifetime survival of discordant pairs under the live patrol
($n = 1024$, pooled over ten seeds).  The first sweep absorbs the
uniform waiting ramp of the balance argument; what survives it decays
geometrically, sweep after sweep, the escape-and-retry cascade whose
absence from the rigorous side is exactly the content of
Conjecture~\ref{conj:balance}.}
\label{fig:lifetail}
\end{figure}
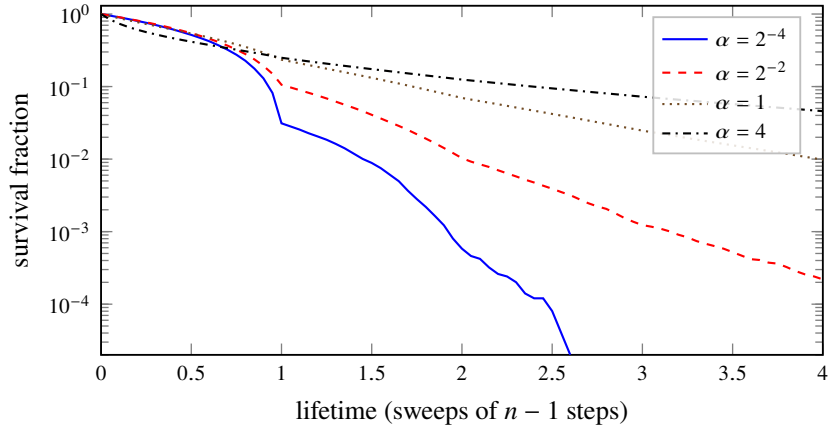

The audit confirms the conjecture's mechanism, not merely its number.
At $\alpha = 2^{-4}$ births occur at $0.97\,\alpha$ per step, the mean
lifetime is $0.525$ sweeps against the ideal $0.5$, $97\%$ of
discordances die on the first cursor visit, and the Little product
matches the measured $K^\ast$ to $0.2\%$; the standing stock is
$0.508\,\alpha n$.  As the rate grows the queue thickens in a specific,
measurable way: by $\alpha = 4$ only $58\%$ of drift events create new
discordances (the rest land on already-discordant locations and cancel),
the mean lifetime stretches to $0.93$ sweeps, the drift itself performs
$72\%$ of the removals, and the two distortions cancel almost exactly,
holding $K^\ast/(\alpha n)$ inside $[0.508,\, 0.554]$ across the grid,
at $n = 1024$ and $n = 4096$ alike.  Little's law holds to three digits
in every cell, which is the strongest internal-consistency check the
model offers: the equilibrium is exactly a queue in balance, its inflow
and sojourn measured separately and multiplying back to its size.

\subsection{Ages, widths, coverage}\label{sec:exp-cert}

Table~\ref{tab:ages} is the equity ledger.  The patrols convert their
deterministic cursor into a deterministic age, and the age into a uniform
certified width: \emph{every} entry of the board is certified to $\pm 8$
positions at the $95\%$ level, with realized widths $\pm 6$, independent
of $n$ (Remark~\ref{rem:width}).  Repeated insertion sort, the better
error minimizer, has no deterministic age bound; its realized worst age is
$85\%$ larger.  The random baseline decouples age from budget entirely and
its widths grow without bound.  Per-item probe load makes the same point
in the language of fragile complexity~\cite{AfsFagHam19}: the patrol
caps every item's comparison load at $2$ per cycle by construction, while
the baseline's load is coupon-collector skewed.

\begin{table}[H]
\centering\small
\caption{Staleness and its price at $n = 4096$, $\alpha = 1$,
$\delta = 0.05$.  Maximum age is reported in units of $n$; $D$ is evaluated
at the realized age and at the deterministic patrol bound $2(n-1)$.}
\label{tab:ages}
\begin{tabular}{lcccc}
\toprule
maintainer & $F^\ast/n$ & realized max age & $D$ at realized age
           & $D$ at bound\\
\midrule
cyclic patrol        & $1.005$ & $1.00\,n$ & $\pm 6$ & $\pm 8$\\
boustrophedon patrol & $0.961$ & $1.00\,n$ & $\pm 6$ & $\pm 8$\\
repeated insertion   & $0.927$ & $1.85\,n$ & $\pm 8$ & ---\\
random adjacent      & $1.768$ & $4.50\,n$ & $\pm 11$ & unbounded\\
\bottomrule
\end{tabular}
\end{table}

Table~\ref{tab:coverage} tests the promises.  Motion-only coverage is
$1.000$ at every gap and level: Lemma~\ref{lem:tail} is sound with room
to spare, and the room is the price of Freedman generality (the certified
radius runs $2.5$--$3\times$ the realized $99$th percentile).  The
operational rows close the loop on the full reported interval of
Definition~\ref{def:operational}: with zero padding the patrol's
certificates already cover at $99.7\%$ against a $95\%$ nominal level,
because the motion slack absorbs the residual estimate error (whose
distribution at probe time has mean $0.60$ and $p_{99} = 4$).  The same
slack buys robustness to rate misspecification: radii computed from an
assumed rate $\hat\alpha = \alpha/2$, half the truth, still cover at
$99.95\%$ on the bare drift process at gap $g = n$, and overestimated
rates only widen.  The certificates are honest and modestly priced;
nothing in them is asymptotic.

\begin{table}[H]
\centering\small
\caption{Certificate coverage at $n = 1024$, $\alpha = 1$.  Top block:
motion-only guarantees for gap $g$.  Bottom block: operational coverage under
the live patrol using radius $D(\text{age},\delta)+b$.  All measured values
exceed their nominal levels.}
\label{tab:coverage}
\begin{tabular}{lccccc}
\toprule
& & $\delta = 0.10$ & $\delta = 0.05$ & $\delta = 0.01$
  & $p_{99}|{\Delta \rk}|$\\
\midrule
motion, $g = n/4$ & $D$ / coverage & $4$ / $1.000$ & $5$ / $1.000$
  & $6$ / $1.000$ & $2$\\
motion, $g = n$   & $D$ / coverage & $6$ / $1.000$ & $7$ / $1.000$
  & $8$ / $1.000$ & $4$\\
motion, $g = 4n$  & $D$ / coverage & $10$ / $1.000$ & $11$ / $1.000$
  & $13$ / $1.000$ & $7$\\
\midrule
operational, $b=0$ & coverage & --- & $0.9973$ & $0.9992$ & \\
operational, $b=2$ & coverage & --- & $0.9996$ & $0.9999$ & \\
operational, $b=4$ & coverage & --- & $0.9999$ & $1.0000$ & \\
\bottomrule
\end{tabular}
\end{table}

The two-part contract is audited by the split-seed protocol of
Proposition~\ref{prop:calibration}.  Fifteen calibration runs at
$n = 1024$, $\alpha = 1$ set the padding to $b = 4$ at target residual
exceedance $10^{-2}$ and estimate $\hat\varepsilon_4 = 0.0048$; fifteen
disjoint audit runs then measure held-out exceedance $0.0046$ and
composed-interval coverage $0.99994$ at $\delta = 0.05$.  The
finite-sample Hoeffding certificate at $m = 15$ runs is
$\varepsilon_4 \le 0.32$ at confidence $0.95$: loose, honest about its
$m^{-1/2}$ units, and irrelevant to the proven motion part, which needs
no calibration at all.

\subsection{Shock recovery and the rebuild crossover}\label{sec:exp-shock}

Figure~\ref{fig:shock} runs Theorem~\ref{thm:recovery} and
Proposition~\ref{prop:hybrid} against the machine at $n = 4096$: block
reversals of width $4$ through $4096$ injected into a sorted board
(quiet aftermath, $\alpha = 0$) and into the living equilibrium
($\alpha = 1$), with the pure patrol, the binary-insertion rebuild, and
the swap-count hybrid racing to recover.  At $\alpha = 0$ the
measurement is the theorem with the slack removed: the patrol recovers
in exactly $L$ sweeps at every width ($3.0$, $15.0$, $63.0$, $255.0$,
$1023.0$, $4095.0$ against bounds $L+1$), the rebuild needs $10.0$ to
$11.0$ sweeps regardless, the curves cross at $L \approx 10 = $ the
crossover of Theorem~\ref{thm:recovery}(iii), and the hybrid tracks the
winner on both sides of the crossing, never exceeding its
$2\,\mathrm{min} + 2(n-1)$ budget: $4.0$ sweeps at $w = 4$, between
$22.0$ and $23.0$ after switching.  Nonlocal shocks make the case the
locality model cannot: thirty-two random long-range exchanges produce
$L \approx 3544$, the patrol grinds for $3544$ sweeps while the rebuild
finishes in $10.6$ (Table~\ref{tab:shock-nonlocal}), a factor $334$,
and the hybrid pays only the $T_0$-cycle probation before switching.

\begin{figure}[H]
\centering
\begin{tikzpicture}
\begin{axis}[width=0.49\textwidth, height=6.0cm,
  xmode=log, log basis x=2, ymode=log,
  xlabel={block reversal width $w$}, ylabel={recovery (sweeps)},
  xtick={4,16,64,256,1024,4096},
  xticklabels={$4$,$16$,$64$,$256$,$1024$,$4096$},
  title={\small quiet aftermath ($\alpha=0$)},
  legend pos=north west, legend cell align=left]
\addplot+[mark=*, mark size=1.6pt]
  table[x=w, y=patrol]{figdata/shock_recovery_a0.dat};
\addlegendentry{patrol}
\addplot+[mark=square*, mark size=1.5pt]
  table[x=w, y=resort]{figdata/shock_recovery_a0.dat};
\addlegendentry{rebuild}
\addplot+[mark=triangle*, mark size=1.9pt]
  table[x=w, y=hybrid]{figdata/shock_recovery_a0.dat};
\addlegendentry{hybrid}
\addplot+[no marks, dashed, black]
  table[x=w, y=bound_patrol]{figdata/shock_recovery_a0.dat};
\addlegendentry{$L+1$ (Thm.~\ref{thm:recovery})}
\addplot+[no marks, dotted, black]
  table[x=w, y=bound_resort]{figdata/shock_recovery_a0.dat};
\addlegendentry{$C_{\mathrm{rs}}/(n{-}1)$}
\end{axis}
\end{tikzpicture}\hfill
\begin{tikzpicture}
\begin{axis}[width=0.49\textwidth, height=6.0cm,
  xmode=log, log basis x=2, ymode=log,
  xlabel={block reversal width $w$}, ylabel={recovery (sweeps)},
  xtick={4,16,64,256,1024,4096},
  xticklabels={$4$,$16$,$64$,$256$,$1024$,$4096$},
  title={\small live drift ($\alpha=1$)},
  legend pos=north west, legend cell align=left]
\addplot+[mark=*, mark size=1.6pt]
  table[x=w, y=patrol]{figdata/shock_recovery_a1.dat};
\addlegendentry{patrol}
\addplot+[mark=square*, mark size=1.5pt]
  table[x=w, y=resort]{figdata/shock_recovery_a1.dat};
\addlegendentry{rebuild}
\addplot+[mark=triangle*, mark size=1.9pt]
  table[x=w, y=hybrid]{figdata/shock_recovery_a1.dat};
\addlegendentry{hybrid}
\end{axis}
\end{tikzpicture}
\caption{Recovery after a block-reversal shock at $n = 4096$ (ten seeds,
five at the two largest widths; recovery to the sorted state at
$\alpha = 0$, to $1.25\times$ the pre-shock equilibrium at $\alpha = 1$).
Left: the quiet aftermath realizes Theorem~\ref{thm:recovery} exactly,
recovery $= L$ sweeps, rebuild flat at $11$, crossover at
$L \approx 10$.  Right: under live drift the rebuild's published board
inherits half a generation of staleness and recovery costs roughly
$100$ sweeps even for small shocks, so the patrol dominates at every
width; the hybrid's quiet-aftermath trigger misfires here, the regime
the next subsection's diagnostics are built to detect.}
\label{fig:shock}
\end{figure}
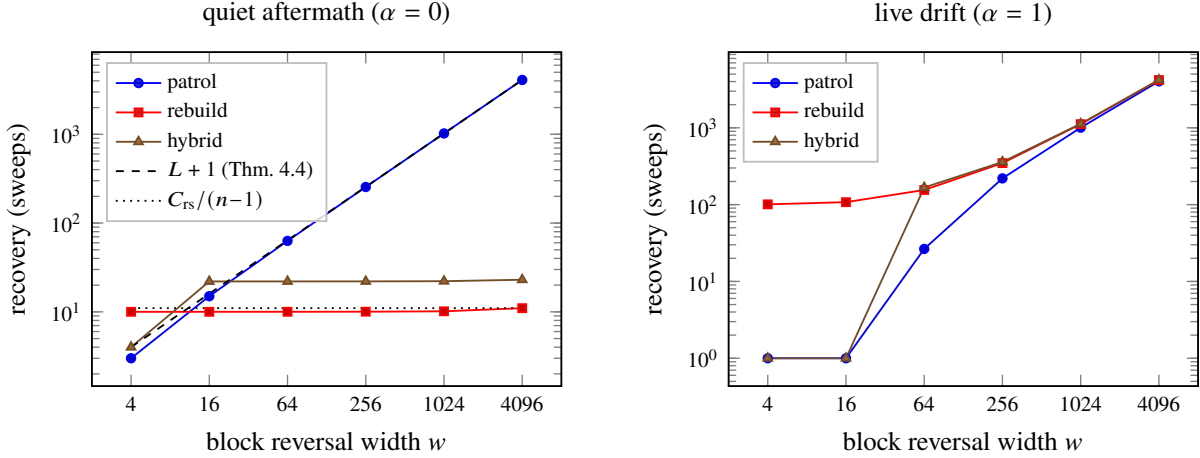

\begin{table}[H]
\centering\small
\caption{Nonlocal shocks: thirty-two uniformly random long-range
exchanges at $n = 4096$ (ten seeds; sweeps to recovery).  The induced
overstatement $L$ is $\Theta(n)$ and the crossover verdict flips to the
rebuild at $\alpha = 0$; under live drift both responses are slow and
the difference compresses.}
\label{tab:shock-nonlocal}
\begin{tabular}{lcccc}
\toprule
& $L$ & patrol & rebuild & hybrid\\
\midrule
$\alpha=0$ & $3544$ & $3544$ & $10.6$ & $22.2$\\
$\alpha=1$ & $3462$ & $3047$ & $3554.0$ & $3566.0$
\\
\bottomrule
\end{tabular}
\end{table}

Two honest boundaries emerge.  The crossover and the hybrid's
competitive guarantee are quiet-aftermath statements, and the
$\alpha = 1$ panel shows both why they matter and where they stop: a
rebuild that must comparison-shop against a moving truth publishes a
stale board, so continuous repair wins every local contest under live
drift, while a genuinely nonlocal shock still defeats any local cursor.
The policy variable is therefore not the drift rate but the
\emph{overstatement geometry} of the change, which is observable: the
swap counter detects that repair is not converging, and the displacement
diagnostics of Section~\ref{sec:exp-ea} estimate $L$ directly from
re-evaluation panels.

\subsection{The frontier under drift}\label{sec:exp-maxima}

The frontier experiments run on the maintenance model directly, two
drifting orders and two patrols, and share nothing with the
evolutionary benchmarks: they are an algorithms experiment in their own
right.  They are designed to keep three different
quantities from blurring into one another, and they are reported in that
order: the \emph{safety certificate} (the proven bound, evaluated on the
same snapshots it protects), the \emph{empirical tightness} (how far
realized error sits below each bound, and how the gap moves with
geometry), and the \emph{adversarial construction} (the family on which
the bound is attained, run rather than merely described).

Table~\ref{tab:maxima} and Figure~\ref{fig:snapshot} put
Theorem~\ref{thm:maxima} in the plane.  At equilibrium each patrol carries
$\approx 0.55\,n$ discordances, so the charging bound permits about
$2.3 \times 10^3$ frontier errors; the measured figure is $0.48$ per
snapshot, because the frontier of a random instance passes through
$\Theta(\log n)$ points and only discordances incident to that thin
staircase can hurt it.  The certificate is doing exactly its job (it is
instance-wise,
distribution-free, and met within its factor $2$ on the antidiagonal
family, where \emph{every} point is on the frontier); the experiment
quantifies the distance between that worst case and a generic one at
three to four orders of magnitude.  The snapshot makes the locality
visible: two hundred discordances are in the picture and only two of them
are anywhere the staircase can feel.

\begin{table}[H]
\centering\small
\caption{Evolving maxima, $n = 1024$, $\alpha = 1$, two cyclic patrols,
$30$ seeds $\times\, 60$ snapshots, initial coordinate orders independent.}
\label{tab:maxima}
\begin{tabular}{lcccc}
\toprule
& $|M|$ & $|M \diff \widehat M|$ & bound $2(K^x{+}K^y)$ & ratio\\
\midrule
mean & $7.6$ & $0.48 \pm 0.19$ & $2.26 \times 10^3$ & $2.0\times10^{-4}$\\
\bottomrule
\end{tabular}
\end{table}

Table~\ref{tab:frontier-local} is the tightness read, with the
adversarial construction as its last row: it prices the localized bound
of Theorem~\ref{thm:maxima-local} across the geometric spectrum, the two
coordinate orders coupled by a Gaussian copula from independence through
the exactly antidiagonal family.  The localized certificate
$2K_{\mathrm{loc}}$ tracks the realized frontier error within a factor
of $2.3$ \emph{everywhere}: $1.1$ against a realized $0.49$ on
independent orders, $89.9$ against $39.8$ on the antidiagonal family
where every point is maximal, while the global budget $2(K^x{+}K^y)$
stands fixed near $2.3\times 10^3$ throughout.  The refinement costs one
pass over $M \cup \widehat M$ to evaluate, is exactly as
distribution-free as the original, and converts the frontier theorem
from a safety statement that is loose by three to four orders of
magnitude on generic instances into one that is loose by a factor near
its own charging constant.

\begin{table}[H]
\centering\small
\caption{The localized frontier certificate across front geometries
($n = 1024$, $\alpha = 1$, thirty seeds, twelve snapshots each;
coordinate orders coupled by a Gaussian copula with the stated
correlation).  All three columns bound or measure
$|M \diff \widehat M|$ per snapshot.}
\label{tab:frontier-local}
\begin{tabular}{lcccc}
\toprule
front geometry & $|M|$ & realized & $2K_{\mathrm{loc}}$
  & $2(K^x{+}K^y)$\\
\midrule
independent & $6.7$ & $0.49$ & $1.1$ & $2258$\\
$\rho=-0.5$ & $11.6$ & $0.59$ & $1.2$ & $2258$\\
$\rho=-0.9$ & $28.1$ & $1.49$ & $3.1$ & $2258$\\
$\rho=-0.99$ & $70.1$ & $8.21$ & $17.3$ & $2258$\\
antidiagonal & $167.7$ & $39.79$ & $89.9$ & $2258$
\\
\bottomrule
\end{tabular}
\end{table}

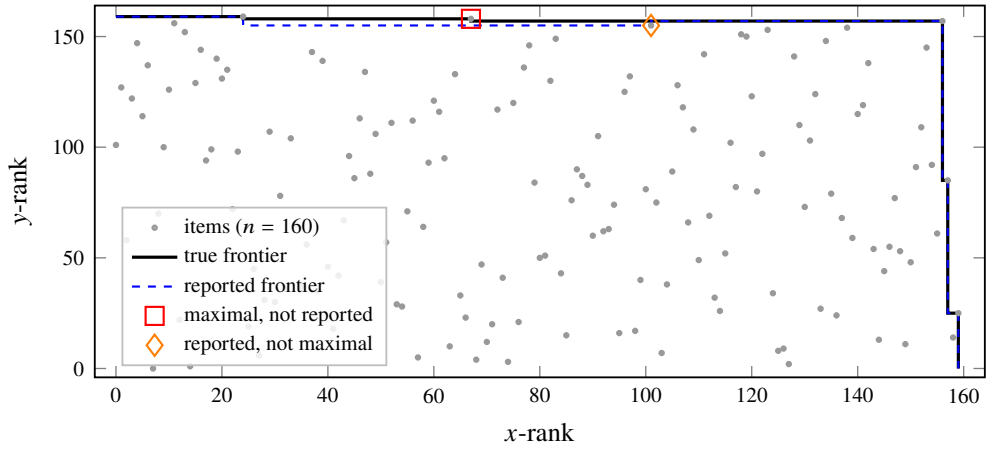
\begin{figure}[H]
\centering
\begin{tikzpicture}
\begin{axis}[width=0.84\textwidth, height=6.5cm,
  xlabel={$x$-rank}, ylabel={$y$-rank},
  xmin=-4, xmax=164, ymin=-4, ymax=164,
  legend pos=south west, legend cell align=left]
\addplot[only marks, mark=*, mark size=0.85pt, color=black!40]
  table[x=x, y=y]{figdata/snap_points.dat};
\addlegendentry{items ($n=160$)}
\addplot[no marks, very thick, color=black]
  table[x=x, y=y]{figdata/snap_true_st.dat};
\addlegendentry{true frontier}
\addplot[no marks, thick, dashed, color=blue]
  table[x=x, y=y]{figdata/snap_est_st.dat};
\addlegendentry{reported frontier}
\addplot[only marks, mark=square, mark size=3.4pt, thick, color=red]
  table[x=x, y=y]{figdata/snap_fn.dat};
\addlegendentry{maximal, not reported}
\addplot[only marks, mark=diamond, mark size=4pt, thick, color=orange]
  table[x=x, y=y]{figdata/snap_fp.dat};
\addlegendentry{reported, not maximal}
\end{axis}
\end{tikzpicture}
\caption{One snapshot of the drifting plane ($n = 160$, $\alpha = 1$,
after $30$ sweeps; items drawn at their current hidden rank coordinates).
Despite $K^x=102$ and $K^y=99$, only one maximal point is missed (square) and
one non-maximal point is reported (diamond): $|M\diff\widehat M|=2$ versus a
worst-case bound of $402$.}
\label{fig:snapshot}
\end{figure}

\subsection{Selection under drift}\label{sec:exp-selection}

Table~\ref{tab:selection} and Figure~\ref{fig:selection} run
Section~\ref{sec:selection} against the machine.  Three reads, in
parallel with the frontier.  First, boundary locality: the patrol
misplaces $0.5$--$1.2$ candidates per boundary across $k = 4$ to $512$,
flat in $k$ to within noise, while the transfer bound stands at $46$;
exactly as for the staircase, the bound spends a unit on every
discordance and random instances put almost none of them where a boundary
can feel it.  Second, the price of generational re-evaluation: at the
identical budget of one comparison per step, the re-sorter's published
board carries $K^\ast = 2.05\,n$, a factor $3.7$ over the patrol, because
its information is on average half a generation stale; the factor passes
through the tournament identity unchanged ($4.0 \times 10^{-3}$ against
$1.08 \times 10^{-3}$ per duel) and shows up as a factor three to four in
truncation error at every $k$.  Continuous certified maintenance beats
periodic re-evaluation at equal cost, and the margin is the staleness.
Third, certified elitism does what Proposition~\ref{prop:elitism}
promises and is not vacuous: with the uniform margin $D + b \le 12$ the
rule certifies $50\%$ of the boundary capacity at $k = 16$, $88\%$ at
$k = 64$, and $96\%$ at $k = 256$, and across all $559{,}800$ certified
picks it commits one error.  The guarantee costs the picks nearest the
boundary and nothing else.

\begin{table}[H]
\centering\small
\caption{Selection under drift at $n = 1024$, $\alpha = 1$ ($30$ seeds
$\times\,60$ snapshots).  The top block compares truncation and tournament
error at equal budget; the bottom block reports certified-elitism yield and
precision for $\delta=0.05$, $b=4$.}
\label{tab:selection}
\begin{tabular}{lcccccc}
\toprule
& & \multicolumn{3}{c}{$|T_k \diff \widehat T_k|$} & bound & tournament\\
\cmidrule(lr){3-5}
maintainer & $K^\ast/n$ & $k{=}16$ & $64$ & $256$
  & $2\lfloor\sqrt{K}\rfloor$ & error\\
\midrule
cyclic patrol  & $0.553$ & $0.61$ & $0.71$ & $0.90$ & $46.6$
  & $1.08\times 10^{-3}$\\
generational   & $2.064$ & $2.17$ & $2.36$ & $2.76$ & $90.6$
  & $4.03\times 10^{-3}$\\
\midrule
certified, patrol & & $k{=}16$ & $64$ & $256$ & picks & wrong\\
\cmidrule(lr){3-5}
yield     & & $0.50$ & $0.88$ & $0.96$ & $559{,}800$ & $1$\\
precision & & $1.000$ & $1.000$ & $1.000$ & & \\
\bottomrule
\end{tabular}
\end{table}

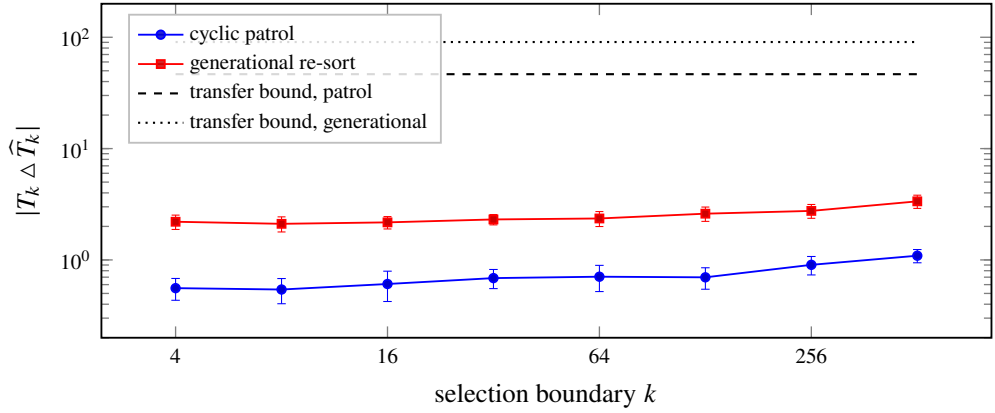
\begin{figure}[H]
\centering
\begin{tikzpicture}
\begin{axis}[width=0.84\textwidth, height=6.0cm,
  xmode=log, log basis x=2, ymode=log,
  xlabel={selection boundary $k$},
  ylabel={$|T_k \diff \widehat T_k|$},
  xtick={4,16,64,256},
  xticklabels={$4$,$16$,$64$,$256$},
  ymin=0.2, ymax=200,
  legend pos=north west, legend cell align=left]
\addplot+[mark=*, mark size=1.6pt, error bars/.cd, y dir=both, y explicit]
  table[x=k, y=sym, y error=sd]{figdata/sel_cyclic.dat};
\addlegendentry{cyclic patrol}
\addplot+[mark=square*, mark size=1.5pt, error bars/.cd, y dir=both,
  y explicit] table[x=k, y=sym, y error=sd]{figdata/sel_generational.dat};
\addlegendentry{generational re-sort}
\addplot+[no marks, dashed, thick, black]
  table[x=k, y=cyc]{figdata/sel_bound.dat};
\addlegendentry{transfer bound, patrol}
\addplot+[no marks, dotted, thick, black]
  table[x=k, y=gen]{figdata/sel_bound.dat};
\addlegendentry{transfer bound, generational}
\end{axis}
\end{tikzpicture}
\caption{Truncation selection error against the boundary $k$ ($n = 1024$,
$\alpha=1$, thirty seeds, log--log).  Patrol error is three to four times
lower than generational re-sorting at equal budget; dashed lines show the
worst-case transfer certificates of Theorem~\ref{thm:topk}.}
\label{fig:selection}
\end{figure}

\subsection{Non-Poisson drift, shocks, and frontier density}
\label{sec:exp-robust}

The baseline model is now stressed in four directions.  A compound-Poisson
process groups crossings into local bursts of mean six; a hot-spot process
places $90\%$ of events in an $8\%$ band of rank locations; a regime process
alternates rates $1/8$ and $4$ once per sweep; and explicit changes reverse a
contiguous rank block or exchange random nonlocal positions.  All continuous
tests use $n=4096$, thirty paired seeds, and the same mean unit drift budget.
For shocks we record the immediate Kendall increase, the area above the
pre-change equilibrium, and sweeps required to return within $25\%$ of that
equilibrium.  A detected shock invalidates pre-change motion certificates;
coverage is resumed only for candidates probed after the detector fired.

Figure~\ref{fig:robust} shows why a single scalar ``drift rate'' is
insufficient.  Concentrating crossings can reduce global Kendall error because
the maintainer repeatedly revisits the damaged region, whereas alternating
slow and fast regimes leave a larger residual stock.  The shock inequality of
Proposition~\ref{prop:shock} holds in every run, but recovery depends on where
the inversions land.  These experiments are intentionally outside the
Poisson/uniform certificate theorem; they test graceful degradation, not
nominal coverage.

\begin{figure}[H]
\centering
\begin{tikzpicture}
\begin{axis}[width=0.80\textwidth,height=6.0cm,
  ybar,bar width=8pt,ylabel={$K^\ast/n$},
  symbolic x coords={poisson,compound,hotspot,regime},
  xtick=data,xticklabel style={rotate=20,anchor=east},
  legend pos=north west]
\addplot table[x=model,y=cyclic]{figdata/stress_drift.dat};
\addlegendentry{cyclic patrol}
\addplot table[x=model,y=insertion]{figdata/stress_drift.dat};
\addlegendentry{repeated insertion}
\end{axis}
\end{tikzpicture}
\caption{Rank-maintenance error under matched mean event budgets but different
temporal and spatial structure ($n=4096$, thirty paired seeds).  The result
separates event rate from burstiness and location concentration.}
\label{fig:robust}
\end{figure}
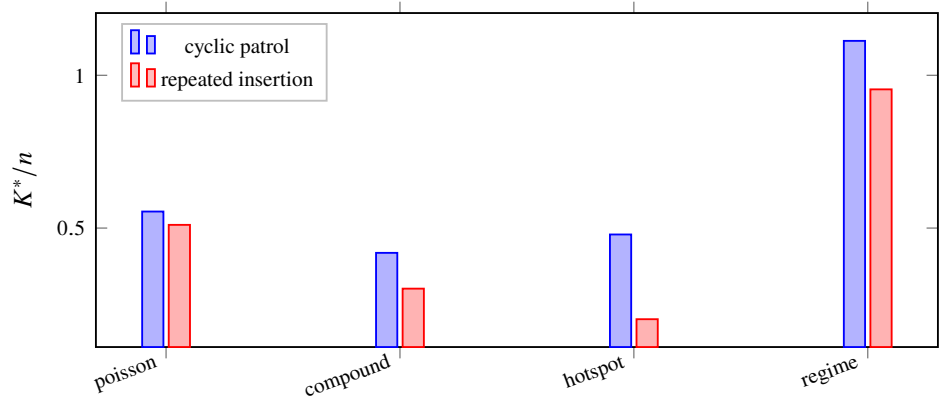

The Pareto experiment independently varies geometric difficulty.  The two
coordinate orders are generated by a Gaussian copula with correlations
$0,-0.5,-0.9,-0.99$ and by the exactly antidiagonal construction.  Negative
correlation thickens the nondominated set from $\Theta(\log n)$ toward all
$n$ candidates.  This closes the gap left by the random snapshot: the charging
theorem remains a safety certificate in every regime, while its empirical
tightness rises as discordances become increasingly likely to touch the
frontier.  The fully antidiagonal family remains the adversarial endpoint.

\subsection{Full dynamic evolutionary loops}\label{sec:exp-ea}

We next place evaluation management inside an optimizer.  Every claim in
this subsection lives in one of the paper's three layers, and each table
names its register.  \emph{Rank-process maintenance} (Kendall disorder,
ages, certified yield and precision) and \emph{selection-error transfer}
(tournament decision error, certified-pick precision) are the registers
where the theory speaks and the conclusions are sharp.  \emph{Optimizer
performance} (regret, recovery) couples the maintained ranking to the
landscape and the variation operators; conclusions there are empirical
by design, comparative, interval-bounded with paired tests and
rank-biserial effect sizes, and never extrapolated beyond the two
benchmark families.  The common
steady-state EA uses a population of $\mu=512$, binary-tournament parent
selection, bit mutation at rate $1/d$ for Dynamic BitMatching, Gaussian
mutation for Moving Peaks, and tournament replacement.  Dynamic BitMatching
uses $d\in\{256,1024\}$ and target changes of $1\%,5\%,20\%$; Moving Peaks
uses five dimensions, ten peaks, and shift severities $1,3,6$.  Environments
change every $5000$ objective evaluations, either abruptly or progressively.
Each run receives $15{,}000$ evaluations, spanning three scheduled changes.

The main maintenance allocation is $25\%$, with $10\%$ and $50\%$ sensitivity
runs.  Patrol pays two scalar evaluations for every comparison.  Periodic
full re-evaluation publishes after a complete pass; random partial and
elite-first policies refresh one candidate at a time; the sentinel policy
samples eight fixed candidates and launches a full pass after detecting a
change; the no-re-evaluation policy spends the entire budget on offspring.
The oracle refreshes the population without charge and is therefore an upper
control, not a feasible competitor.  All non-oracle policies share the same
total objective-evaluation budget.  A second policy set widens the
comparison beyond re-evaluation scheduling to the classic
dynamic-optimization responses, charged identically: random immigrants
(diversity maintenance), triggered hypermutation, memory with reinjection
of archived elites after a detected change, partial restart of the worse
half, and the patrol+refresh hybrid of
Proposition~\ref{prop:hybrid} lifted to the evaluation layer, which
patrols continuously and answers a detected change with one full
re-evaluation pass.

For every condition and seed we report time-averaged and $90$th-percentile
regret, recovery evaluations, Kendall disorder, binary-tournament decision
error, certified top-$64$ yield and precision, and the exact number of
maintenance and offspring evaluations.  Means carry bootstrap $95\%$
intervals.  Patrol--baseline comparisons use paired Wilcoxon tests with Holm
correction and matched rank-biserial effect sizes.  The tests address
evaluation management rather than claiming that this deliberately simple EA
is competitive with specialized dynamic optimizers.

Finally, we audit the abstraction itself.  On a fixed panel of $512$
candidates, each environmental change is converted to its induced rank
permutation.  We measure Kendall burst size, the Fano factor across changes,
the normalized entropy of affected rank locations, and the fraction of
candidates moving more than one rank.  Low-burst, local processes are the
regime represented by the adjacent model.  Large nonlocal fractions explain
when a global re-evaluation pass can outperform local patrol despite having a
worse age distribution.  This diagnostic is the bridge from a concrete
fitness landscape to the assumptions of the maintenance theory.

\begin{table}[H]
\centering\footnotesize
\caption{Representative abrupt-change EA results at the main $25\%$
maintenance allocation: mean regret with bootstrap $95\%$ interval,
and binary-tournament error.  BitMatching uses $d=1024$, severity $5\%$;
Moving Peaks uses severity $3$.  The free-refresh oracle is not feasible.}
\label{tab:dynamic-ea}
\setlength{\tabcolsep}{3.8pt}
\begin{tabular}{lcccc}
\toprule
& \multicolumn{2}{c}{BitMatching} & \multicolumn{2}{c}{Moving Peaks}\\
policy & regret $[95\%]$ & sel. error & regret $[95\%]$ & sel. error\\
\midrule
patrol & $444.97$ $[442.48,447.37]$ & $0.155$ & $17.64$ $[15.11,20.11]$ & $0.113$\\
periodic full & $444.36$ $[441.73,446.92]$ & $0.091$ & $17.69$ $[15.20,20.15]$ & $0.055$\\
random partial & $444.70$ $[442.31,446.93]$ & $0.112$ & $18.67$ $[15.99,21.29]$ & $0.052$\\
elite-first & $443.54$ $[441.29,445.75]$ & $0.070$ & $17.68$ $[15.24,20.07]$ & $0.039$\\
sentinel-triggered & $443.24$ $[440.99,445.44]$ & $0.065$ & $17.64$ $[15.18,20.02]$ & $0.034$\\
no re-evaluation & $441.17$ $[437.26,445.11]$ & $0.107$ & $16.39$ $[13.96,18.73]$ & $0.118$\\
oracle (free refresh) & $434.17$ $[431.83,436.39]$ & $0.011$ & $16.34$ $[13.91,18.71]$ & $0.010$
\\
\bottomrule
\end{tabular}
\end{table}

\begin{table}[H]
\centering\small
\caption{Induced rank-process diagnostics on fixed candidate panels of
$512$.  Burst is Kendall shock size; entropy is normalized across rank
locations; nonlocal is the fraction moving more than one rank;
$p_{95}$ disp.\ is the $95$th percentile of per-candidate rank
displacement per change, the panel-level estimate of the overstatement
$L$ of Theorem~\ref{thm:recovery}.}
\label{tab:rank-diagnostics}
\begin{tabular}{lccccc}
\toprule
landscape & burst & Fano & entropy & nonlocal & $p_{95}$ disp.\\
\midrule
BitMatching, gradual & $18755.7$ & $26.0$ & $0.984$ & $0.970$ & $143$\\
BitMatching, abrupt & $18755.7$ & $26.0$ & $0.984$ & $0.970$ & $143$\\
Moving Peaks, gradual & $2670.9$ & $84.9$ & $0.984$ & $0.816$ & $21$\\
Moving Peaks, abrupt & $2670.9$ & $84.9$ & $0.984$ & $0.816$ & $21$
\\
\bottomrule
\end{tabular}
\end{table}

Table~\ref{tab:dynamic-ea} rejects a simplistic winner narrative.  At the
main $25\%$ allocation, patrol is slightly more accurate than no
re-evaluation on Moving Peaks, but less accurate on the highly nonlocal
BitMatching shocks; periodic, elite-first, and sentinel-triggered refresh are
more accurate in both representative abrupt conditions.  Meanwhile the
evaluations spent on every maintenance policy reduce the number of offspring,
so no re-evaluation and the free-refresh oracle attain lower regret here.  The
$10\%/25\%/50\%$ sensitivity runs make the trade-off explicit: increasing
patrol budget reduces its Kendall disorder and tournament error, but can
increase optimization regret.  Patrol is therefore useful when bounded
per-item staleness and auditable recency are requirements, not as a universal
replacement for allocating evaluations to search.

Table~\ref{tab:rank-diagnostics} explains the outcome, and the recovery
theory now turns the explanation into a rule.  Both benchmark panels
produce high-entropy, predominantly nonlocal rank shocks, far outside a
single adjacent event per evaluation; the displacement column makes the
verdict quantitative.  On a panel of $\mu = 512$ the crossover of
Theorem~\ref{thm:recovery} sits at $L^\ast = C_{\mathrm{rs}}(512)/511 - 1
\approx 7$: BitMatching changes displace the $95$th-percentile candidate
by $143$ ranks and Moving Peaks by $21$, both far beyond $L^\ast$, so the
theory itself prescribes global refresh, or the hybrid, over pure patrol
on these landscapes, and the regret table agrees.  The prescription
inverts where the diagnostic does: the maintenance-layer experiments of
Section~\ref{sec:exp-shock} sit below the crossover under gradual drift
and local change, exactly where the patrol's certified repair wins.  The
practical reading of the pair of tables is a three-way rule: measure the
per-change displacement quantiles on a re-evaluation panel; below
$L^\ast$ run the patrol, far above it trigger a refresh, and in between,
or when the regime is unknown, run the hybrid, whose probation cost is
bounded by Proposition~\ref{prop:hybrid}.  In the reset experiment,
serving pre-change motion certificates degrades coverage as shock size
grows (to $0.9695$ for the largest block), whereas withholding them until
a candidate is re-probed restores measured motion coverage to $1.000$ at
the half-sweep checkpoint, with $50\%$ yield.  This is precisely the
operational failure mode hidden by a homogeneous gradual-drift benchmark.

Table~\ref{tab:ea-strong} widens the comparison to the standard
dynamic-optimization responses, all under the same budget accounting:
random immigrants, triggered hypermutation, memory with reinjection of
archived elites, partial restart, and the patrol+refresh hybrid that
patrols continuously, folds sentinel probes into its budget, and answers
a detected change with one full re-evaluation pass.  On raw regret the
severity sweep orders the policies exactly as a budget argument
predicts.  At $1\%$ severity every maintenance policy loses to spending
nothing: no-re-evaluation wins outright, and the hybrid trails it with
Holm-adjusted $p = 8.9\times10^{-8}$ and rank-biserial effect $1.0$,
because ten flipped bits are cheaper to outrun with offspring than to
audit.  At $20\%$ severity the ordering inverts: no-re-evaluation falls
to the bottom third, partial restart and sentinel-triggered refresh lead,
and the hybrid significantly beats periodic full refresh
($p = 0.027$).  No search-side response separates from the simple
sentinel elsewhere; the paired intervals overlap.  Where the policies
separate sharply is the selection layer: on BitMatching the hybrid cuts
the patrol's tournament decision error from $0.155$ to $0.121$, and on
Moving Peaks from $0.113$ to $0.055$, matching periodic refresh while
keeping bounded per-item staleness and the certificate interface between
changes.  This is the displacement rule read off the machine: both
panels sit far above the crossover ($p_{95}$ displacements of $143$ and
$21$ against $L^\ast \approx 7$), so pure patrol is the weakest
selection layer here, exactly as Theorem~\ref{thm:recovery} prices it,
and the value it retains is the part the regret column cannot see,
namely certified, auditable selection between changes.  The conclusion
the table supports is the paper's, not a stronger one: as an
optimizer-side response the patrol is unremarkable, and as a certified
selection-maintenance layer it is cheap, auditable, and composable with
exactly the refresh mechanisms that dominate when changes are nonlocal.

\begin{table}[H]
\centering\footnotesize
\caption{Stronger dynamic baselines at the main allocation (abrupt
changes; BitMatching $d = 1024$, severity $5\%$; Moving Peaks severity
$3$; thirty seeds; mean regret with bootstrap $95\%$ intervals
and binary-tournament decision error).  Register: optimizer performance,
with the selection-error columns reporting the maintenance register
inside the same runs.}
\label{tab:ea-strong}
\setlength{\tabcolsep}{3.8pt}
\begin{tabular}{lcccc}
\toprule
& \multicolumn{2}{c}{BitMatching} & \multicolumn{2}{c}{Moving Peaks}\\
policy & regret $[95\%]$ & sel. error & regret $[95\%]$ & sel. error\\
\midrule
patrol & $444.97$ $[442.49,447.48]$ & $0.155$ & $17.64$ $[15.26,20.13]$ & $0.113$\\
hybrid patrol+refresh & $444.75$ $[442.49,446.83]$ & $0.121$ & $18.31$ $[15.78,20.94]$ & $0.055$\\
periodic full & $444.36$ $[441.73,446.90]$ & $0.091$ & $17.69$ $[15.34,20.09]$ & $0.055$\\
sentinel-triggered & $443.24$ $[441.01,445.40]$ & $0.065$ & $17.64$ $[15.35,20.02]$ & $0.034$\\
no re-evaluation & $441.17$ $[437.14,445.05]$ & $0.107$ & $16.39$ $[14.00,18.79]$ & $0.118$\\
random immigrants & $443.71$ $[440.75,446.76]$ & $0.119$ & $18.49$ $[15.71,21.35]$ & $0.103$\\
triggered hypermutation & $443.19$ $[440.24,446.23]$ & $0.101$ & $18.15$ $[15.81,20.61]$ & $0.104$\\
memory + reinjection & $445.31$ $[442.10,448.46]$ & $0.114$ & $17.00$ $[14.75,19.31]$ & $0.083$\\
partial restart & $446.10$ $[443.42,448.64]$ & $0.121$ & $18.24$ $[15.72,20.80]$ & $0.057$
\\
\bottomrule
\end{tabular}
\end{table}

Table~\ref{tab:ea-ablation} closes the protocol with ablations along
the four axes a budgeted evaluation manager actually controls or
suffers: change frequency (the interval between abrupt changes), the
maintenance share of the evaluation budget, tournament selection
pressure, and population size, each varied around the main setting on
the representative BitMatching condition.  Longer change intervals reduce
both regret and decision error throughout the table.  Increasing the
maintenance share from $0.10$ to $0.50$ lowers decision error for every
active maintenance policy but raises regret because fewer evaluations remain
for offspring; the no-maintenance row is unchanged by construction.  Larger
tournaments improve regret for all four policies without uniformly improving
decision error, while the smaller population improves regret at the fixed
total evaluation budget.

\begin{table}[H]
\centering\footnotesize
\caption{Ablations on BitMatching ($d = 1024$, severity $5\%$, abrupt,
thirty seeds): mean regret\,/\,tournament decision error per cell.  The
main setting is interval $5{,}000$, share $0.25$, binary tournaments,
$\mu = 512$; each block varies one axis.  Register: optimizer
performance left of the slash, selection-error transfer right.}
\label{tab:ea-ablation}
\setlength{\tabcolsep}{4.5pt}
\begin{tabular}{lcccc}
\toprule
& patrol & hybrid & sentinel & none\\
\midrule
interval $2{,}500$ & $450.3$\,/\,$0.229$ & $451.9$\,/\,$0.198$ & $450.4$\,/\,$0.146$ & $452.8$\,/\,$0.221$\\
interval $5{,}000$ (main) & $445.0$\,/\,$0.155$ & $444.7$\,/\,$0.121$ & $443.2$\,/\,$0.065$ & $441.2$\,/\,$0.107$\\
interval $10{,}000$ & $440.3$\,/\,$0.106$ & $440.4$\,/\,$0.097$ & $439.8$\,/\,$0.035$ & $432.4$\,/\,$0.060$\\
\midrule
share $0.10$ & $442.6$\,/\,$0.175$ & $441.4$\,/\,$0.164$ & $441.0$\,/\,$0.129$ & $441.2$\,/\,$0.107$\\
share $0.50$ & $452.0$\,/\,$0.135$ & $449.8$\,/\,$0.088$ & $452.8$\,/\,$0.033$ & $441.2$\,/\,$0.107$\\
\midrule
tournament $4$ & $429.7$\,/\,$0.147$ & $426.7$\,/\,$0.111$ & $426.4$\,/\,$0.059$ & $424.5$\,/\,$0.126$\\
population $128$ & $395.7$\,/\,$0.136$ & $393.0$\,/\,$0.108$ & $391.2$\,/\,$0.018$ & $404.9$\,/\,$0.140$
\\
\bottomrule
\end{tabular}
\end{table}

\section{Related work}\label{sec:related}

\paragraph{Evolving data.}  The model is the sorting side of the evolving
data framework: Anagnostopoulos, Kumar, Mahdian, and
Upfal~\cite{AnaKumMahUpf11} introduced sorting against one random adjacent
swap per comparison, proved an $\Omega(n)$ steady-state lower bound for
adaptive algorithms, and gave an $O(n \log\log n)$ maintainer; Besa,
Devanny, Eppstein, Goodrich, and Johnson~\cite{BesaDEGJ18,BesaDEGJexp18}
closed the gap by showing repeated insertion sort holds $\Theta(n)$.
Relative to that line, the present paper changes the deliverable: the
maintained object is not a ranking but a ranking \emph{with promises},
meaning deterministic ages, per-item certified intervals, per-item probe
loads, exact stabilization and recovery laws, and transfer theorems into
the frontier and into selection.  The lower-bound side is re-proved twice
in forms that yield explicit constants at every drift rate: by
Poissonization and parity decoupling for the oblivious class
(Theorem~\ref{thm:lower}), and by a dead-zone argument for the
location-oblivious class that contains the patrols themselves
(Theorem~\ref{thm:floor}), so the structure under study is squeezed
between a proven floor and a conjectured ceiling whose constant the
lifetime ledger measures from inside.  Table~\ref{tab:positioning}
fixes the comparison against the reference points in one place.

\begin{table}[tb]
\centering\footnotesize
\caption{The patrol against its reference points, one comparison per
step unless stated.  Steady error at $\alpha = 1$; ages in steps;
recovery from a shock of overstatement $L$ under quiet aftermath.
Kinetic structures solve a different problem (known motion plans, exact
event-driven certificates) and anchor the certificate vocabulary rather
than compete on this budget.}
\label{tab:positioning}
\setlength{\tabcolsep}{4.0pt}
\begin{tabular}{lccccc}
\toprule
structure & steady error & worst age & per-query promise
  & recovery & space\\
\midrule
comparison patrol
  & $\Theta(n)$: floor proven, & $2(n-1)$ & motion radius $+$
  & $(L-1,\,L]$ cycles, & $3n{+}O(1)$\\
\quad (here)
  & const.\ $\tfrac12$ conjectured & deterministic
  & calibrated residual & exact & \\
repeated insertion~\cite{BesaDEGJ18}
  & $\Theta(n)$ proven & $1.85\,n$ measured, & none & not stated
  & $3n{+}O(1)$\\
  & ($0.51\,n$ measured) & no det.\ bound & & & \\
evolving sorting~\cite{AnaKumMahUpf11}
  & $O(n\log\log n)$ & not tracked & none & not stated & $O(n)$\\
periodic rebuild
  & $2.06\,n$ measured & $n\lceil\log_2 n\rceil$ & none
  & $C_{\mathrm{rs}}(n)$ exact, & $4n{+}O(1)$\\
\quad (binary insertion)
  & (staleness-driven) & by construction & & $L$-independent & \\
kinetic structures~\cite{BasGuiHer99}
  & exact between events & event-driven & exact certificate
  & by event & model-\\
  & (needs flight plans) & & & processing & dependent\\
\bottomrule
\end{tabular}
\end{table}

\paragraph{Self-stabilization and queues.}  Started from arbitrary
corruption, the patrol converges to the legitimate state without
intervention, the defining property of self-stabilizing
protocols~\cite{Dijkstra74}, and Theorem~\ref{thm:stabilize} pins its
convergence time to within one cycle.  The hybrid's rent-or-buy structure
is the elementary trade of competitive analysis~\cite{KarManRudSle88},
with the patrol renting, the rebuild buying, and the swap counter as the
meter.  The equilibrium audit treats the standing disorder as a queue and
verifies $L = \lambda W$~\cite{Little61,Stidham74} on exact birth and
death records of individual discordances; we are not aware of a prior
data-structure analysis run at this grain.

\paragraph{Noise versus drift.}  In the noisy-comparison
tradition~\cite{FeiRagPelUpf94} the ground truth is fixed and the oracle
lies; repetition buys certainty.  Here the oracle is exact and the truth
moves; repetition buys nothing but recency, which is why ages, not error
rates, are the structure's currency.

\paragraph{Kinetic structures and patrols.}  Kinetic data
structures~\cite{BasGuiHer99} maintain geometry under \emph{known} motion
plans and audit themselves with certificates that fail at predictable
events.  The displacement certificate of Section~\ref{sec:cert} is the
statistical counterpart when no plan exists.  The patrol itself is the
maintenance-theoretic cousin of patrolling games~\cite{AlpMorPap11} (a
defender walking a cycle against a stochastic rather than strategic
intruder), and the boustrophedon variant borrows its turning pattern from
coverage planning~\cite{ChoPig98}.

\paragraph{Budgets per item.}  Fragile complexity~\cite{AfsFagHam19}
bounds how many comparisons a single element may suffer in a static
computation.  The patrol's per-cycle load of $2$ per item is the dynamic
analogue, and Table~\ref{tab:ages} measures what abandoning such a bound
costs downstream of the data structure, in certificate width.

\paragraph{Evolutionary computation.}  Dynamic optimization is a mature
subfield of evolutionary computation~\cite{Branke02,NguYanBra12}: moving
optima, change detection, and re-evaluation policies are its standard
concerns.  Dynamic BitMatching continues the theoretically studied dynamic
OneMax line~\cite{Droste03}; Moving Peaks is the standard continuous tracking
benchmark~\cite{MosChi13,Branke20}.  Runtime analysis of dynamic
optimization studies how the frequency and magnitude of change govern
trackability~\cite{RohLehYao09} and proves regime separations for
concrete heuristics on dynamic fitness families~\cite{LisWit16}; the
displacement diagnostics and the crossover of
Theorem~\ref{thm:recovery} supply, from the maintenance side, exactly
the magnitude statistic that axis parameterizes, together with the
per-decision selection error rates such analyses consume.  Those
literatures typically modify the
search heuristic through memory, immigrants, diversity, or prediction.  We
instead isolate evaluation management inside an otherwise fixed EA and ask
what a comparison-limited evaluator can certify about its fitness ranking.
Theorem~\ref{thm:topk} and
Propositions~\ref{prop:tournament} and~\ref{prop:elitism} are
maintenance-layer guarantees for truncation, tournament, and elitist
selection under drift, and Theorems~\ref{thm:maxima}
and~\ref{thm:maxima-local} are the corresponding guarantees for the
nondominated set that multi-objective methods of the
NSGA line~\cite{DebPraAgaMey02} carry as their elite.  The theory of
randomized search heuristics~\cite{DroJanWeg06,DoeNeu20} analyzes the
same elementary mutation kernel as an \emph{optimizer}; here it is the
\emph{adversary}, and Wilson's mixing bound~\cite{Wilson04} calibrates
how fast it forgets.  The full loops of Section~\ref{sec:exp-ea} compare
continuous patrol with periodic, random, elite-first, and event-triggered
re-evaluation, and with the field's own response repertoire, random
immigrants, triggered hypermutation, memory, and restart, at equal
objective-evaluation budgets.  The empirical rank diagnostics expose
locality, burst size, and displacement quantiles as the variables
determining which policy is appropriate, and Theorem~\ref{thm:recovery}
turns the displacement diagnostic into a crossover with a proof rather
than a heuristic.

\paragraph{Adjacent context.}  Selecting comparisons under a budget
against drifting strengths is the maintenance face of dueling
bandits~\cite{YueBroKleJoa12}, with the objective shifted from regret to
structural error.  Maxima under
explicit updates go back to~\cite{KungLucPre75,OveLee81}, and sliding
windows over data streams maintain skylines under visible arrivals and
expiries~\cite{TaoPap06}; the reduction of
Theorems~\ref{thm:maxima} and~\ref{thm:maxima-local} needs no update
stream at all, only the two Kendall distances, which is what separates
this frontier layer from kinetic maxima (planned motion, event queues),
from dynamic maxima under explicit insertions and deletions, and from
streaming skylines (Remark~\ref{rem:geo-model}).
Adaptive sorting~\cite{EstWood92} bounds static sorting cost by the
input's existing disorder; the stabilization law of
Theorem~\ref{thm:stabilize} is its maintenance-side analogue, with the
overstatement $L$ playing the role of the presortedness measure and the
bound exact rather than asymptotic.  Background on inversions and bubble
passes is classical~\cite{Knuth98}, and the metric inequalities are
Diaconis--Graham~\cite{DiaGra77}.

\section{Concluding remarks}\label{sec:conclusion}

Selection under drifting fitness has an information problem before it has
a search problem, and the object that solves it is a data structure whose
true output is trust: not the ranking, which is necessarily wrong
somewhere, but exact statements about where and by how much.  This paper
builds that structure and closes its accounts.  The lower bounds say what no schedule can avoid,
twice: for oblivious probing by parity decoupling, and for the patrol's own
location-oblivious class by a dead-zone argument, so the measured
equilibrium stands on a proven floor.  The bump lemma gives the patrol an
exact self-stabilization law, measured to within one probe of its
prediction, and from it a recovery calculus whose crossover at
$L \approx \log_2 n$ decides, with a proof, when local repair should yield
to a global rebuild; a swap-count hybrid obeys the decision without being
told the shock.  Certificates are served as a two-part contract that never
mixes what is proven with what is calibrated, and the equilibrium constant
that remains conjectural is audited at the level of individual discordance
lifetimes, where Little's law closes to three digits.

As the evaluation-management layer beneath rank-based selection, the
structure carries its boundary conditions inside the theory rather than
as caveats around it.  Local gradual rank motion is
the regime where certified continuous repair wins, and the displacement
diagnostics locate any landscape relative to the crossover; the benchmark
landscapes fall on the nonlocal side, where the same theorems prescribe
refresh or the hybrid, and the equal-budget loops, now including random
immigrants, hypermutation, memory, and restart, confirm the prescription
in both directions.  The resulting contribution is neither a universally
superior optimizer nor a relabelled sorting algorithm: it is a certified
maintenance layer with explicit floors, exact recovery laws, a falsifiable
equilibrium constant, and a decision rule for its own retirement.

\paragraph{Reproducibility.}  The ancillary directory \texttt{anc/}
contains the drift engines, maintainers, the inversion ledger, dynamic
landscapes, full EA loop, selection layer, shock and hybrid policies,
statistical analysis, and closed forms of every stated bound.
Forty-two tests check engine invariants, rank orientation, deterministic
theorems and tightness constructions, the bump lemma and stabilization
windows, recovery and hybrid bounds, the equilibrium floor, ledger
exactness, shock accounting, variable-rate certificates, evaluation
budgets, selection-pressure replay, deterministic EA replay, and seeded
stochastic guarantees.  Raw
and summarized CSV files, commands, source hashes, package versions, and
platform information accompany a paper and a quick regeneration profile;
\texttt{README.md} documents both.

\end{document}